\documentclass[a4paper,UKenglish,cleveref, autoref, thm-restate]{lipics-v2021}

\pdfoutput=1 
\hideLIPIcs  

\title{Identifying Tractable Quantified Temporal Constraints within Ord-Horn} 

\author{Jakub Rydval}{Technische Universit\"{a}t Wien, Vienna, Austria}{jakub.rydval@tuwien.ac.at}{https://orcid.org/0000-0002-7961-9492}{This research was funded in whole or in part by the Austrian Science Fund (FWF) [I 5948]. For the purpose of Open Access, the authors have applied a CC BY public copyright licence to any Author Accepted Manuscript (AAM) version arising from this submission.} 

\author{\v{Z}aneta Semani\v{s}inov\'{a}}{Technische Universit\"{a}t Dresden, Dresden, Germany}{zaneta.semanisinova@tu-dresden.de}{https://orcid.org/0000-0001-8111-0671}{The author has been funded by the European Research Council (Project POCOCOP, ERC Synergy
Grant 101071674) and by the DFG (Project FinHom, Grant 467967530). Views and opinions expressed are however
those of the authors only and do not necessarily reflect those of the European Union or the European Research
Council Executive Agency. Neither the European Union nor the granting authority can be held responsible for them.} 

\author{Micha\l\ Wrona}{Jagiellonian University, Krak\'{o}w, Poland}{michal.wrona@uj.edu.pl}{https://orcid.org/0000-0002-2723-0768}{The author is partially supported by National Science Centre, Poland grant number 2020/37/B/ST6/01179.}

\authorrunning{Jakub Rydval, \v{Z}aneta Semani\v{s}inov\'{a}, and Micha\l\ Wrona} 

\Copyright{Jakub Rydval, \v{Z}aneta Semani\v{s}inov\'{a}, and Micha\l\ Wrona} 



\category{} 

\relatedversion{} 

\acknowledgements{The authors thank Dmitriy Zhuk for many inspiring discussions on the topic.}

\nolinenumbers 

\EventEditors{John Q. Open and Joan R. Access}
\EventNoEds{2}
\EventLongTitle{42nd Conference on Very Important Topics (CVIT 2016)}
\EventShortTitle{CVIT 2016}
\EventAcronym{CVIT}
\EventYear{2016}
\EventDate{December 24--27, 2016}
\EventLocation{Little Whinging, United Kingdom}
\EventLogo{}
\SeriesVolume{42}
\ArticleNo{23}

\usepackage{thmtools}
\usepackage{thm-restate}

\usepackage{tabularx} 
\usepackage{amsmath}
\usepackage{array} 
\usepackage{booktabs}
\usepackage{caption}
\usepackage{fancyvrb}
\usepackage{graphicx} 
\usepackage{longtable}
\usepackage{zi4}

 \usepackage{stmaryrd} 
\usepackage{mathtools}

\usepackage[noend,ruled]{algorithm2e} 
\usepackage{enumitem}
\usepackage{microtype} 
\usepackage{xspace}

\usepackage{varwidth}

\newcommand{\case}[2]{\underline{\emph{Case~#1: #2}}}

\newcommand{\Init}{\ensuremath{\S\hspace{1pt}\mathrm{I}}\xspace}  
\newcommand{\Simplify}{\ensuremath{\S\hspace{1pt}\mathrm{S}}\xspace}    
\newcommand{\Trans}{\ensuremath{\S\hspace{1pt}\mathrm{T}}\xspace}  
\newcommand{\AltTrans}{\ensuremath{\S\hspace{1pt}\mathrm{A}}\xspace}    
\newcommand{\Progress}{\ensuremath{\S\hspace{1pt}\mathrm{C}}\xspace}  
\newcommand{\Refute}{\ensuremath{\S\hspace{1pt}\mathrm{R}}\xspace}

\newcommand{\of}[1]{\raisebox{0.25pt}{\textup{{\kern-0.025em{\footnotesize[\raisebox{-0.05pt}{$#1$}}\kern-0.025em{\footnotesize]}\kern0.025em}}}}

\newcommand{\Aut}{\ensuremath{\mathrm{Aut}}\xspace} 
 
\newcommand{\pp}{\pi\pi} 
\newcommand{\dual}{\ensuremath{\mathit{dual}}\,\xspace} 
\newcommand{\mi}{\ensuremath{\mathit{mi}}\xspace} 
\newcommand{\mx}{\ensuremath{\mathit{mx}}\xspace}
\renewcommand{\min}{\ensuremath{\mathit{min}}\xspace}
\renewcommand{\max}{\ensuremath{\mathit{max}}\xspace}
\newcommand{\CSP}{\ensuremath{\mathrm{CSP}}\xspace} 
\newcommand{\NAE}{\ensuremath{\mathrm{NAE}}\xspace}  
\newcommand{\lex}{\ensuremath{\mathit{lex}\xspace}} 
\newcommand{\elel}{\ell\ell} 
\newcommand{\QCSP}{\ensuremath{\mathrm{QCSP}}\xspace}     
\newcommand{\struct}[1]{\mathfrak{#1}}   
\newcommand{\cut}[2]{#1\text{-}#2\text{-cut}}
\newcommand{\varse}{\ensuremath{\mathrm{V}_{\!\exists}}\xspace} 
\newcommand{\varsu}{\ensuremath{\mathrm{V}_{\!\forall}}\xspace} 
\newcommand{\vars}{\ensuremath{\mathrm{V}}\xspace} 
\newcommand{\ass}[1]{\ensuremath{\llbracket#1\rrbracket}\xspace}

\newcommand{\Zaneta}{\textup{\v{Z}}}

\newcommand{\Dima}{\mathrm{D}}
\newcommand{\Michal}{\mathrm{M}^+}
\newcommand{\SMichal}{\mathrm{M}_<^+}
\newcommand{\DMichal}{\mathrm{M}^-}
\newcommand{\DSMichal}{\mathrm{M}_<^-}
\newcommand{\GMichal}{\mathrm{GM}^+}
\newcommand{\GVSMichal}{\mathrm{GVM}_<^+}
\newcommand{\dualGMichal}{\mathrm{GM}^-}
\newcommand{\dualGVSMichal}{\mathrm{GVM}_<^-}
\newcommand{\LessSepGDis}{\mathrm{GSN}}
 \newcommand{\Dis}{\mathrm{SN}} 
\newcommand{\SepDima}{\mathrm{SD}}
\newcommand{\SepMichal}{\mathrm{SM}}
\newcommand{\SepSMichal}{\mathrm{SM}_<}

\newcommand{\LessSepGMichal}{\mathrm{lrGSM}}
\newcommand{\LessSepGSMichal}{\mathrm{lrGSM}_<}
\newcommand{\GreatSepGMichal}{\mathrm{rlGSM}}
\newcommand{\GreatSepGSMichal}{\mathrm{rlGSM}_<}

\begin{document}

\maketitle

\begin{abstract}
The constraint satisfaction problem, parameterized by a relational structure, provides a general framework for expressing computational decision problems. Already the restriction to the class of all finite structures forms an interesting microcosm on its own, but to express decision problems in temporal reasoning one has to take a step beyond the finite-domain realm. 
An important class of templates used in this context are temporal structures, i.e., structures over $\mathbb{Q}$ whose relations are first-order definable using the usual countable dense linear order without endpoints.

In the standard setting, which allows only existential quantification over  input variables, the complexity of finite and temporal constraints has been fully classified. 
In the quantified setting, i.e., when one also allows universal quantifiers, 
there is only a handful of partial classification results and many concrete cases of unknown complexity.
This paper presents a significant progress towards understanding the complexity of the quantified constraint satisfaction problem for temporal structures. 
We provide a complexity dichotomy for quantified constraints over the Ord-Horn fragment, which played an important role in understanding the complexity of constraints both over temporal structures and in Allen's interval algebra.
We show that all problems under consideration are in P or coNP-hard.
In particular, we determine the complexity of the quantified constraint satisfaction problem for $(\mathbb{Q};x=y\Rightarrow x\geq z)$, hereby settling a question open for more than ten years.
\end{abstract} 

\section{Introduction} The constraint satisfaction problem (CSP) of a structure $\struct{B}$ in a finite relational signature $\tau$, denoted by $\CSP(\struct{B})$, is the problem of deciding whether a given primitive positive $\tau$-sentence holds in $\struct{B}$.
The class of all \emph{finite-domain} CSPs, i.e., where $\struct{B}$ can be chosen finite, famously constitutes a large fragment of NP that admits a dichotomy between P and NP-completeness~\cite{zhuk2020proof}. 
Quantified constraint satisfaction problems (QCSPs) generalize CSPs by allowing both existential \emph{and} 
universal quantification over input variables.
The complexity of such problems is much less understood already for finite structures, 
 the state of the art being a complexity classification for QCSPs of finite structures with all unary relations and three-element structures with all singleton unary relations~\cite{zhuk2022qcsp}.
For infinite structures, the investigations essentially follow the CSP programme, which was initiated by the study of the CSPs of structures over $\mathbb{N}$ (or $\mathbb{Q}$) whose relations are definable by Boolean combinations of equalities and disequalities, the so-called \emph{equality structures}~\cite{bodirsky2008complexity}. 
The full complexity classification for quantified equality constraints was completed quite recently~\cite{zhuk2023complete}, by resolving the long-standing question of determining the complexity of $\QCSP(\mathbb{Q};\Dima)$, where \[\Dima\coloneqq \{(x,y,z)\in \mathbb{Q}^3 \mid x=y \Rightarrow x=z\}.\] 
 This question was left open in~\cite{qecsps}, where all the remaining results have been provided.
The next in line are \emph{temporal structures}, which are structures with domain $\mathbb{Q}$ whose relations are first-order definable over $\{<\}$, where $<$ interprets as the usual unbounded dense linear order.
The relations of such structures are called \emph{temporal}.

By definition, temporal structures form a richer class than equality structures.
While the complexity of temporal CSPs has been classified more than a decade ago~\cite{bodirsky2010complexity}, there is only a handful of partial classification results regarding the complexity of temporal QCSPs~\cite{ToTheMax, charatonik2008quantified,charatonik2008tractable,chen2012guarded,WronaMFCS14}. 
Yet, already from this limited amount of available data it is apparent that the majority of the pathological cases is concentrated in the \emph{Ord-Horn} (OH) fragment, we elaborate on this below.
The OH fragment comprises all temporal structures whose relations are definable by an OH formula, i.e., a conjunction of clauses of the form 
\begin{align}
   (x_1\neq y_1 \vee \cdots \vee x_k\neq y_k \vee x_{k+1} \geq y_{k+1})  \label{eq:OrdHorn}
\end{align}
 for $k\geq 0$, where the last disjunct is optional and some variables might be identified~\cite{Book}. 

\subsection{Ord-Horn}

OH was first introduced and used by Nebel and B\"urckert to describe a maximally tractable constraint language containing all basic relations on Allen's interval algebra~\cite{nebel1995reasoning}. 
  For a full classification of maximally tractable subalgebras of Allen, see~\cite{KrokhinJeavonsJonsson}. 
  In the context of CSPs over temporal structures,  OH is not even a maximally tractable language  
   as it is properly contained in two of the nine maximally tractable fragments identified in~\cite{bodirsky2010complexity}:  $\min$, $\max$, $\mx$, $\dual\mx$,  $\mi$, $\dual\mi$, $\elel$, $\dual\elel$, and a constant operation.
   We remark that all of the above are in fact operations on $\mathbb{Q}$. The \emph{dual} of an operation $f$ on $\mathbb{Q}$ is the operation $(x_1,\dots,x_n) \mapsto -f(-x_1,\dots,-x_n)$~\cite{bodirsky2010complexity}, e.g., $\max$ is the dual of $\min$. 
  The question which of the nine fragments are also maximal w.r.t.\ tractability of the QCSP was investigated in~\cite{ToTheMax}, 
  and answered positively in the first four cases.
  The answer is negative in the last three cases~\cite{chen2011quantified,zhuk2023complete}, and the question remains open for $\mi$ and $\dual\mi$. 
  In the intersections of $\elel$ with $\mi$ and with $\dual\mi$ lie the OH structures $(\mathbb{Q};\Michal)$ and $(\mathbb{Q};\DMichal)$, respectively, where
  \[
         \Michal\coloneqq \{(x,y,z)\in \mathbb{Q}^{3} \mid x=y \Rightarrow x\geq z\} \quad \text{and} \quad 
         \DMichal\coloneqq \{(x,y,z)\in \mathbb{Q}^{3} \mid x=y \Rightarrow x\leq z\}. 
  \]
  Determining the complexity of $\QCSP(\mathbb{Q};\Michal)$ was posed as an open question in \cite{ToTheMax}; it could have been anywhere between PTIME and PSPACE.  
  Note that its counterpart $\QCSP(\mathbb{Q};\DMichal)$ is essentially the same problem with the order reversed.
  
Apart from temporal structures preserved by a constant operation, OH captures precisely those temporal structures whose CSP is solvable by local consistency checking~\cite{BoRyTCSPs}. 
This well-known generic preprocessing algorithm can be formulated for any CSP satisfying some reasonable structural assumptions~\cite{bodirsky2013datalog}, and thus OH constraints are fairly well understood from the CSP perspective.
However, the analysis of OH constraints in the quantified setting requires a surprisingly large amount of creativity. As a simple example, already $\QCSP(\mathbb{Q};R)$ for the OH relation $R$ defined by $(x_1\neq x_2 \vee x_3\geq x_4)\wedge \phi$~is
 in P if $\phi$ equals $(x_3\geq x_1) \wedge (x_1\geq x_3)  \wedge (x_3\neq x_4)$~\cite{chen2012guarded}, 
  coNP-complete if $\phi$ equals $(\bigwedge_{i,j\in \{1,2\}}  x_{i} \neq x_{j+2})$~\cite{ZhukExample},  
 and PSPACE-complete if $\phi$ is the empty conjunction~\cite{zhuk2023complete}.

The set of \emph{Guarded Ord-Horn} (GOH) formulas~\cite{chen2012guarded} is defined inductively.
In the base case we are allowed to take OH formulas of the form $(x\leq y)$, $(x_1\neq y_1 \vee \cdots \vee x_k\neq y_k)$, or $(x\neq x_1 \vee \cdots \vee x\neq x_k)\vee (x<y) \vee (y\neq y_1 \vee \cdots \vee y\neq y_{\ell})$.
In the induction step we can form formulas of the form $\psi_1\wedge \psi_2$ or $(x_1\leq y_1 \vee \cdots \vee x_k\leq y_k)\wedge (x_1\neq y_1 \vee \cdots \vee x_k\neq y_k \vee \psi)$, where $\psi,\psi_1,\psi_2$ 
are GOH formulas.
%
Thus, newly added disequalities are guarded by atomic $\{\leq\}$-formulas. 
A GOH structure may only contain temporal relations definable by GOH formulas. 
Observe that the tractable template from the previous paragraph is GOH.  
\begin{theorem}[\cite{chen2012guarded}] \label{thm:GOH_tractable} Let $\struct{B}$ be a GOH structure. Then $\QCSP(\struct{B})$ is in PTIME.
\end{theorem}
The tractability result from~\cite{chen2012guarded} is conceptually simple and based on pebble games generalizing local consistency methods.
 At the same time, all quantified OH constraints outside of GOH are coNP-hard or admit a LOGSPACE reduction from $\QCSP(\mathbb{Q};\Michal)$. 
 \begin{theorem}[\cite{WronaMFCS14}]
\label{thm:noGOH}
    Let $\struct{B}$ be  an  OH  structure. Then  one of the following holds.
    \begin{itemize} 
    \item $\struct{B}$ is GOH.
       \item $\QCSP(\struct{B})$ is coNP-hard.
        \item $\struct{B}$ primitively positively defines $\Michal$ or $\DMichal$. 
    \end{itemize} 
    \end{theorem}
  There was a prospect that $\QCSP(\mathbb{Q};\Michal)$ would be PSPACE-hard, because the PSPACE-hardness proof from \cite{zhuk2023complete} for $\QCSP(\mathbb{Q};\Dima)$, when adjusted appropriately, almost yields a proof of PSPACE-hardness for this QCSP.
  In that case, Theorems~\ref{thm:GOH_tractable} and~\ref{thm:noGOH} would immediately yield a dichotomy between P and coNP-hardness for quantified OH constraints.
 However, it turns out that the situation is more complicated, as we explain below.

 \subsection{Contributions}
 On the one hand, we prove tractability for the entire $\pp$-fragment of quantified OH constraints, which in particular includes $\QCSP(\mathbb{Q};\Michal)$.  
 Here by $\pp$ we refer to the ``\emph{projection-projection}'' operation from \cite{bodirsky2010complexity}, which played an important role in identifying the maximally tractable temporal CSP languages covered by $\min$, $\mi$ and $\mx$. 
  
\begin{theorem}  \label{thm_ptime}
If $\struct{B}$ is an OH structure preserved by $\pp$, then $\QCSP(\struct{B})$ is in PTIME.
\end{theorem}
  The proof of Theorem~\ref{thm_ptime} stretches over the entirety of Section~\ref{section:proof_of_theorem1}, and is the main technical contribution of the present paper.
  In a certain sense, our algorithm for QCSPs of OH structures preserved by $\pp$ 
  also generalizes local consistency methods. 
  We iteratively expand a given instance $\Phi$ of $\QCSP(\struct{B})$ by constraints associated to relations whose arity is bounded by the size of $\Phi$ and which have short primitive positive definitions in $\struct{B}$, until a fixed-point is reached.
  The condition for the expansion by these constraints is tested using an oracle for $\CSP(\struct{B},<)$.
  The algorithm is thus not very far from the well-known framework of Datalog with existential rules~\cite{abiteboul1995foundations,ceri2012logic}.
  We believe that it will also prove useful in identifying the complexity of quantified temporal constraints outside of OH, e.g., in the case of $\mi$ or $\pp$.

 On the other hand, we confirm that $\QCSP(\mathbb{Q};\Michal)$ indeed walks a very fine line between tractability and hardness.
We show that, if $\Michal$ is combined with any OH relation $R$ that is not preserved by $\pp$,
then the resulting QCSP becomes coNP-hard, even if $\QCSP(\mathbb{Q};R)$ is tractable.
 Intuitively, either $(\mathbb{Q};\Michal,R)$ already primitively positively defines $\Dima$ and we use the PSPACE-hardnees proof from~\cite{zhuk2023complete} directly, or we replace each constraint of the form $\Dima(x,y,z)$ in the proof by $\Michal(x,y,z)\wedge \Michal(z,z,x)$. 
The latter, however, is not entirely conditional, and certain issues arise due to the transitivity of $\leq$.
These issues can be partially (but not entirely) resolved using constraints associated to 
\[
\Zaneta \coloneqq \{ (x_1,y_1,x_2,y_2)\in \mathbb{Q}^4 \mid (x_1 \neq y_1 \vee x_2 \neq y_2)\wedge (y_1 < y_2)\},
\]
 which is quantified primitively positively definable in $(\mathbb{Q};\Michal,R)$, ultimately leaving us with a proof of coNP-hardness.

 \begin{theorem}  \label{thm_hardness}
  Let $\struct{B}$ be an OH structure. If $\struct{B}$ is not GOH, not preserved by $\pp$, and not preserved by  $\dual\pp$, then $\QCSP(\struct{B})$ is coNP-hard.
  \end{theorem}
  Theorem~\ref{thm_hardness} is proved in Section~\ref{section:second_proof}, by a careful combination of syntactic pruning arguments, Theorem~\ref{thm:noGOH}, and a new coNP-hardness proof inspired by the PSPACE-hardness proof from~\cite{zhuk2023complete}. 

By combining Theorem~\ref{thm:GOH_tractable}, Theorem~\ref{thm_ptime} together with its analogue for $\dual\pp$, 
and Theorem~\ref{thm_hardness}, we get the following dichotomy for quantified OH constraints.   
 \begin{theorem} \label{conj:GOH_vs_PP_vs_HARD}
  Let $\struct{B}$ be an OH structure. 
  Then $\QCSP(\struct{B})$ is solvable in polynomial time if $\struct{B}$ is GOH, preserved by $\pp$, or preserved by $\dual\pp$. Otherwise, $\QCSP(\struct{B})$ is coNP-hard. 
  \end{theorem}
 Omitted proofs or their parts can be found in Appendix~\ref{appendix:proofs}.  
 
\section{Preliminaries}
 
\subsection{First-order structures} 
The set $\{1,\dots,n\}$ is denoted by $[n]$.
 In the present paper, we consider structures $\struct{A} = (A; R_1, \ldots, R_k)$ over a finite relational signature $\tau$. For the sake of simplicity, we often use the same symbol $R$ for both the relation $R^{\struct{A}}$ and the relational symbol $R$.   
%
%
An \emph{expansion} of $\struct{A}$ is a $\sigma$-structure $ \struct{B}$ with $A=B$ such that $ \tau\subseteq \sigma$ and $R^{\struct{B}}=R^{\struct{A}}$ for each  $R\in \tau$.  
We write $(\struct{A},R)$ for the expansion of $\struct{A}$ by the relation $R$ over $A$.
An $n$-ary \emph{polymorphism} of a relational structure $\struct{A}$ is a mapping $f\colon A^{n}\rightarrow A$ such that, for every $k$-ary relation symbol $R\in \tau$ and tuples $\bar{t}_{1},\dots,\bar{t}_{n}\in R^{\struct{A}}$,  we have $f(\bar{t}_1,\dots,\bar{t}_n)\in R^{\struct{A}}$.
%
%
We say that $f$ \emph{preserves} $\struct{A}$ to indicate that $f$ is a polymorphism of $\struct{A}$. We might also say that an operation \emph{preserves} a relation $R$ over $A$ if it is a polymorphism of $(A;R)$.   
An \emph{endomorphism} is a unary polymorphism.
 
We assume that the reader is familiar with classical first-order logic; we allow the first-order formulas $x=y$ and $\bot$ (the nullary falsity predicate).
Let $T$ be a set of first-order $\tau$-sentences over a common signature $\tau$ and $\phi,\psi$ $\tau$-formulas whose free variables are among $\bar{x}$. 
We say that  $\phi$ \emph{entails} $\psi$ w.r.t.\ $T$ if $
\struct{A}\models \forall \bar{x} (\phi \Rightarrow \psi)$ holds for all models $\struct{A}$ of $T$.
We do not explicitly mention $T$ if it is clear from the context, e.g., the theory of linear orders.
A first-order $\tau$-formula $\phi$ is \emph{primitive positive} (pp) if it is of the form $\exists x_{1},\dots,x_{m}  (\phi_{1}\wedge \dots \wedge \phi_{n})$, where each $\phi_{i}$ is \emph{atomic}, i.e., of the form $\bot$, $(x_{i}=x_{j})$, or $R(x_{i_{1}},\dots,x_{i_{\ell}})$ for some $R\in \tau$. 
\emph{Quantified primitive positive} (qpp) formulas generalize pp-formulas by allowing both existential \emph{and} universal quantification. 
If $\phi$ and $\psi$ are (q)pp-formulas, then $\phi \wedge \psi$ can be rewritten into an equivalent (q)pp-formula, so we treat such formulas as (q)pp-formulas as well. 

The \emph{(quantified) constraint satisfaction problem} for a structure $\struct{B}$, denoted by $\mathrm{(Q)}\CSP(\struct{B})$, is the computational problem of deciding whether a given (q)pp $\tau$-sentence holds in $\struct{B}$. 
By \emph{constraints}, we refer to the conjuncts in the quantifier-free part of a given (Q)CSP instance of $\mathrm{(Q)}\CSP(\struct{B})$.
In the QCSP framework,  we usually think of an instance as a game between two
players: an \emph{existential player} (EP) and a \emph{universal player} (UP) who assign values to the existentially and universally quantified variables, respectively.
To every moment of the game we associate a partial function $\ass{\cdot}$ from the variables into the domain of the parametrizing structure describing values assigned to the variables by either of the players.
The instance is true if and only if the EP has a winning strategy in this game, i.e., can respond to all moves of the UP while keeping all constraints satisfied. Otherwise, the instance is false and the UP has a winning strategy, i.e., can violate a constraint regardless of the moves of the EP. 

If $\struct A$ is a $\tau$-structure and $\phi(x_1,\dots,x_n)$ is a $\tau$-formula with free variables $x_1,\dots,x_n$, then the relation 
	$\{(a_1,\dots,a_n)\in A^{n} \mid \struct{A}\models \phi(a_1\dots, a_n) \}$ is the \emph{relation defined by $\phi$ in $\struct A$}, and denoted by $\phi^{\struct{A}}$. 
Let $S$ be a set of $\tau$-formulas.
 We say that a relation $R$ has a \emph{$S$-definition} in $\struct{A}$, or that $\struct{A}$ \emph{$S$-defines} $R$,  if $R$ equals $\phi^{\struct{A}}$ for some   $\phi \in S$.
For instance, $S$ can be the set of all quantifier-free or primitive positive formulas over $\tau$.
    We might also say that a relation $R$ $S$-defines another relation $R'$ if the structure $(A;R)$ $S$-defines $R'$.  
The next proposition is  folklore in the constraint satisfaction literature.  
\begin{proposition}[\cite{Book, qecsps}]  \label{prop:reductions}  Let $\struct{A}$ be a  structure and $R\subseteq A^k$ for some $k\in \mathbb{N}$.  
\begin{itemize} 
    \item If $R$ has a pp definition in $\struct{A}$, then it is preserved by all  polymorphisms of $\struct{A}$. 
    \item If $R$ is (q)pp-definable in $\struct{A}$, then $(\mathrm{Q})\CSP(A;R)$ reduces to $(\mathrm{Q})\CSP(\struct{A})$ in LOGSPACE.
\end{itemize} 
\end{proposition}       

\subsection{Temporal structures}
The following operations play an essential role in the study of temporal (Q)CSPs.
\begin{definition}  \label{def:basic_operations}   
Let $e_{<0},e_{>0}$ be any fixed endomorphisms of $(\mathbb{Q};<)$ satisfying $e_{<0}(x)<0$ and $e_{>0}(x)>0$ for every $x\in \mathbb{Q}$.
Moreover, let $\lex$ be any fixed binary operation on $ \mathbb{Q}$ satisfying $\lex(x,y) <\lex(x',y')$ iff $x<x'$, or $x=x'$ and $y<y'$ for all $x,x',y,y'\in \mathbb{Q}$.
We denote by $\pp$\footnote{In contrast to previous literature on temporal CSPs,
we deviate from the  
notation pp from \cite{bodirsky2010complexity} that  
clashes with the shortcut for ``primitive positive'' and use $\pp$ instead.} and $\elel$ the binary operations on $\mathbb{Q}$ defined by 
\[\pp(x,y) = \begin{cases}
      e_{<0} x & x \leq  0, \\
      e_{>0}y & x > 0, 
\end{cases} \qquad \elel(x,y) = \begin{cases}
      \lex(e_{<0}x,e_{<0}y) & x \leq  0, \\
      \lex(e_{>0}y,e_{>0}x) & x > 0
\end{cases}
\]
(see Figure~\ref{fig:pp_and_ll} for an illustration of the weak linear orders induced by $\pp$ and $\elel$; undirected and directed edges represent constant and increasing values, respectively).  
\end{definition}  
\begin{figure}
     \centering
      \includegraphics[width=0.7\linewidth]{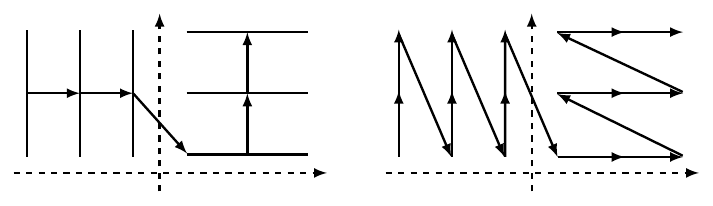} 
     \caption{A visualization from~\cite{bodirsky2010complexity} of $\pp$ (left) and $\elel$ (right). }
     \label{fig:pp_and_ll}
\end{figure}  
Note that  $\lex$ and $\elel$ are injective.
 
Since $(\mathbb{Q};<)$ has quantifier-elimination~\cite{Hodges}, every temporal relation is in fact quantifier-free-definable in $(\mathbb{Q};<)$; 
We may further assume that every quantifier-free definition is in \emph{conjunctive normal form} (CNF).
  We might sometimes refer to temporal relations directly using their CNF-definitions.
 Also, it will sometimes be convenient to work with formulas over the structure $({\mathbb Q};\leq,\neq)$ instead of the structure $({\mathbb Q};<)$, e.g., in the definition of OH from the introduction.
 We will frequently use the following algebraic characterization of OH. 
\begin{proposition}[\cite{ll,BoRyTCSPs}]
\label{prop:ordhorn} The following are equivalent for a temporal relation $R$.
\begin{itemize}
    \item $R$ is OH;
    \item $R$ is preserved by $\elel$ and $\dual\elel$;
    \item $R$ is preserved by every injective polymorphism of $(\mathbb{Q};\leq)$. 
\end{itemize}
\end{proposition}
Recall the operation $\pp$ from Definition~\ref{def:basic_operations}.
Temporal relations preserved by $\pp$ admit the following syntactic description.
\begin{proposition}[\cite{ToTheMax}]  \label{ppsynt}	A temporal relation is preserved by $\pp$ if and only if it has a CNF-definition over $\{\neq,\geq\}$ where, for each clause $\psi$, there are $k,\ell \geq 0$ such that $\psi$ is of the~form 
  \begin{align}
      (x\neq y_1\vee \cdots \vee x\neq y_k \vee x \geq z_1 \vee \cdots \vee x \geq z_{\ell}).  \label{eq:CNF}
  \end{align}  
\end{proposition}  
A syntactic description of temporal relations preserved by $\dual\pp$ can be obtained simply by replacing each instance of $\geq$ by $\leq$.
We hereby have both a syntactic and an algebraic description for OH and ($\dual$) $\pp$, and can use them interchangeably. 
From the syntactic descriptions \eqref{eq:OrdHorn} and \eqref{eq:CNF} it is apparent that the fragments OH and $\pp$ are incomparable.
Their intersection consists of all temporal relations definable by a conjunction of clauses of the form~\eqref{eq:CNF} where $\ell\in \{0,1\}$.
This can be shown using the following lemma.
\begin{restatable}{lemma}{ElimMin}   \label{lemma:ElimMin}
     Let $R$ be an OH relation defined by a quantifier-free formula $\phi $ in CNF over the signature $\{\leq,\neq\}$ containing a clause $\psi_1 \vee \psi_2$, where $\psi_1$ is equivalent to $(x \geq z_1 \vee \cdots \vee x \geq z_{\ell})$ for some variables $x$ and $z_1, \dots, z_{\ell}$. Then we may replace $\psi_1$ in $\phi $ by $(x \geq z_i)$ for some $i \in [\ell]$ so that the resulting formula still defines $R$.
\end{restatable} 

In the present article, the intersection of OH and ($\dual$) $\pp$ is the sole source of all newly identified tractable QCSPs.
It is convenient to work with a finite relational basis.
Recall the relations $\Michal$ and $\DMichal$ from the introduction.
 By Lemma~\ref{lemma:ElimMin}, Proposition~\ref{prop:reductions}, Proposition~\ref{prop:ordhorn}, and Lemma~\ref{lemma:pp-def} below, a temporal relation is OH and preserved by $\pp$ if and only if it is pp-definable in $(\mathbb{Q};\Michal,\neq)$.
  An analogous statement holds for $\dual\pp$ and $(\mathbb{Q};\DMichal,\neq)$.  
\begin{restatable}{lemma}{PPDef}    \label{lemma:pp-def} Let $\phi_{k}(x,y_1,\dots,y_{k},z_1)$ be a pp-definition of the $(k+2)$-ary temporal relation defined by $\mu_k\coloneqq (x\neq y_1\vee \cdots \vee x\neq y_k \vee x \geq z)$. Then  
  $\exists h \ \phi_{k}(x,y_1,\dots,y_{k},h)  \wedge \Michal(h,h,x) \wedge \Michal(h,y_{k+1},z_1)$
  is a pp-definition of  the $(k+3)$-ary temporal relation defined by $\mu_{k+1}$.
 \end{restatable} 
  %

\section{Identifying the tractable case(s)\label{section:proof_of_theorem1}}  
This section is devoted to the proof of Theorem~\ref{thm_ptime}.
An analogue of Theorem~\ref{thm_ptime}  for $\dual\pp$ can be proved simply by reversing the order in each individual statement.
We mainly prove that $\QCSP(\mathbb{Q};\Michal)$ is in PTIME, using Algorithm~\ref{algo:dima}. 
Given this fact, the theorem can be proved easily.   
We remark that, in Algorithm~\ref{algo:dima},  instances of $\QCSP(\mathbb{Q};\Michal)$ are viewed as sentences over the signature $\{\geq,\neq\}$ in prenex normal form whose quantifier-free part is in CNF.  

We first need to introduce some terminology.
   In the Section~\ref{section:proof_of_theorem1}, $\Phi$ always denotes an arbitrary or explicitly specified instance  of $\QCSP(\mathbb{Q};\Michal)$; we denote its quantifier-free part by $\phi$, and its variables by $V$.
  Furthermore, we denote the universal variables by $\varsu $ and the existential variables by $\varse $.
  Let $\prec$ be the linear order on all variables of $\Phi$ in which they appear in the quantifier prefix of $\Phi$.
When we write $A\prec B$ for $A,B \subseteq V$, we mean $x\prec y$ for all $x\in A,y\in B$. In particular, this condition is trivially true if one of the two sets is empty.
\begin{definition} \label{def:terminology} For $x,z\in \vars$, we define 
\begin{itemize} 
    \item  $x\equiv z$ if and only if $x,z$ refer to the same variable,
    \item  $x \preceq z$ if and only if $x\equiv z$ or $x\prec z$,
    \item  $x \preceq_{\forall} z$ if and only if $x\equiv z$, or $x\prec z$ and $z\in \varsu$.
\end{itemize} 
For $u\in V$ and $A\subseteq V$, we define
 \begin{itemize}
     \item ${\uparrow}_{u}\coloneqq \{y\in\varsu \mid u\preceq y  \}$,
     \item ${\uparrow}_{A}\coloneqq \bigcup_{u\in A}{\uparrow}_{u}$ (recall that the empty union is empty). 
 \end{itemize}   
\end{definition} 
Note that the three binary relations in Definition~\ref{def:terminology} are transitive.
\begin{definition} \label{def:cut}
    For every pair $x,z \in \vars$, we define $\cut{x}{z}\coloneqq \{u\in \varsu  \mid   \varse\cap \{x,z\} \prec u\}\setminus\{z\}.$
\end{definition}  

Observe that the definition of the $\cut{x}{z}$ depends on how  
$x$ and $z$ are quantified.  
The idea is that $\cut{x}{z}$ represents the universal variables that the UP can always make equal to $x$ to trigger the condition $(x \geq z)$ via an entailed constraint of the form $\big((\bigwedge\nolimits_{v\in A} x=v) \Rightarrow x\geq z\big)$.
Since the UP has full control over the values of these variables with respect to $x$ and $z$, they can be removed from the clauses added in the second last line in Algorithm~\ref{algo:dima}.

Note that, by Lemma~\ref{lemma:pp-def}, the constraints added by Algorithm~\ref{algo:dima} correspond to relations which have pp-definitions in $(\mathbb{Q};\Michal)$ of length linear in their arity.
This means that satisfiability in Algorithm~\ref{algo:dima} can be tested using an oracle for $\CSP(\mathbb{Q};\Michal,<)$, because we can simply replace each constraint by its pp-definition in $\CSP(\mathbb{Q};\Michal)$ while only changing the size of $\Phi$ by a polynomial factor.
%
%
%
%
%
%
%
%
			\begin{algorithm}    
   \SetAlgoLined
       \caption{\label{algo:dima}An algorithm for $\QCSP(\mathbb{Q};\Michal)$.}	  
				\KwIn{An instance $\Phi$ of $\QCSP(\mathbb{Q};\Michal)$ with the quantifier-free part $\phi$}  
				\KwOut{\emph{true} or \emph{false}}   
    \While{$\phi$ changes}{ 
       \For{$x,z,u\in \vars$}{
       \If{$\phi$ contains the clause $(x\geq z)$ or $(z\geq x)$, where $x\prec z$ and $z\in \varsu$}{\Return \emph{false}\;} 
      \If{$\phi\wedge(\bigwedge_{v\in {\uparrow}_u \setminus \{x,z\}} x=v)\wedge (x<z)$ is unsatisfiable}{expand $\phi$ by the clause $\big((\bigwedge_{v\in {\uparrow}_u \setminus (\{x,z\} \cup \cut{x}{z})} x=v)\Rightarrow x\geq z\big)$\;
       }
      }
      }
        \Return \emph{true}\;  
			\end{algorithm}  
\begin{example}\label{example:algo-comp} 
    Consider the instance $\Phi$ of $\QCSP(\mathbb{Q};\Michal)$ defined by 
    \begin{align*} \exists x_1 \forall x_2 \exists x_3   \forall x_4 \exists x_5 \big(& (x_1=x_2\Rightarrow x_1 \geq x_5)  \wedge (x_3=x_2\Rightarrow x_3\geq x_4) \\
        {} \wedge \ & (x_5=x_4\Rightarrow x_5 \geq x_3) \wedge (x_3\geq x_1) \wedge  (x_5\geq x_1)   \big).
    \end{align*} 
   We claim that Algorithm~\ref{algo:dima} derives  $(x_1\geq x_4)$, and thereby rejects on $\Phi$. 
   We first observe that the formula  $\phi\wedge (\bigwedge_{v\in {\uparrow}_u \setminus \{x_1,x_4\}} x_1=v)\wedge (x_1<x_4)$ is satisfiable for every $u\in \{x_1,\dots, x_5\}$.
   Since $x_3, x_5 \in \varse$, it is enough to show that $\phi\wedge (x_1=x_2)\wedge (x_1<x_4)$ is satisfiable, which is witnessed by any assignment satisfying $(x_5=x_1=x_2<x_3<x_4)$.
   On the other hand,  $\phi\wedge (\bigwedge_{v\in {\uparrow}_{x_2} \setminus \{x_1,x_3\}} x_1=v)\wedge (x_1<x_3)$ is not satisfiable.
   Therefore, the algorithm expands $\phi$ by $(x_1=x_2\Rightarrow x_1\geq x_3)$, because $x_4\in \cut{x_1}{x_3}$. 
   But now $\phi\wedge (\bigwedge_{v\in {\uparrow}_{x_2} \setminus \{x_1,x_4\}} x_1=v)\wedge (x_1<x_4)$ is not satisfiable anymore.
   Since $x_2\in \cut{x_1}{x_4}$, the algorithm expands $\phi$ by $(x_1 \geq x_4)$.
\end{example}
 
 As mentioned below Definition~\ref{def:cut}, Algorithm~\ref{algo:dima} rejects correctly because all constraints added during the run of the algorithm are logically entailed by $\Phi$, see Lemma~\ref{lemma:algo_refutes_correctly}.  
\begin{lemma} \label{lemma:algo_refutes_correctly} Suppose that Algorithm~\ref{algo:dima} derives from $\Phi$ a constraint $\psi$. Then $\Phi$ is true if and only if $\Phi$ expanded by $\psi$ is true.  
\end{lemma} 
\begin{proof}  
Denote by $\Psi$ and $\Psi'$ the sentences obtained from $\Phi$ by replacing $\phi$ with $\phi\wedge \psi$ and  $\phi\wedge \psi'$, respectively, where 
\[\psi\coloneqq  \Big(\bigvee\nolimits_{v\in {\uparrow}_u \setminus (\{x,z\} \cup \cut{x}{z})} x\neq v \Big)\vee (x \geq z)
\quad \text{and}\quad 
\psi' \coloneqq \Big(\bigvee\nolimits_{v\in {\uparrow}_u \setminus \{x,z\}} x\neq v\Big) \vee (x \geq z) .
\]
The formula $\phi$ is satisfiable because we can assign the same value to every variable. Since $\phi \wedge \neg \psi'$ is unsatisfiable, we have that 
$\phi$ entails $\psi'$. It follows that $\Phi$ is true iff $\Psi'$ is true.
To complete the proof, we have to show that if $\Psi'$ is true, then $\Psi$ is true. We prove the contraposition and assume that the UP has a winning strategy on $\Psi$. If the UP wins on $\Psi$ by falsifying any clause different from $\psi$, then the very same choices lead the UP to falsifying the same clause in $\Psi'$. Otherwise, the UP falsifies $\psi$ while playing on $\Psi$. Then the UP can play in the same way on $\Psi'$ when it comes to the variables that occur both in $\psi$ and $\psi'$ and set all variables in $\cut{x}{z}$ to the same value as $x$. Note that this is possible since either $x$ is universal or $x$ precedes all variables in $\cut{x}{z}$. It remains to show that $\psi'$ is falsified. Clearly, all $\{\neq\}$-disjuncts are falsified. Since $z$ is either universal or precedes all variables in $\cut{x}{z}$, the disjunct $(x \geq z)$ is falsified as well, because it is falsified in $\psi$. 
\end{proof}

Example~\ref{example:unsat2} showcases how a winning strategy of the UP, obtainable implicitly from Lemma~\ref{lemma:algo_refutes_correctly}, might in fact be uniquely determined.

\begin{figure}[h]
     \centering
      \includegraphics[width=0.75\linewidth]{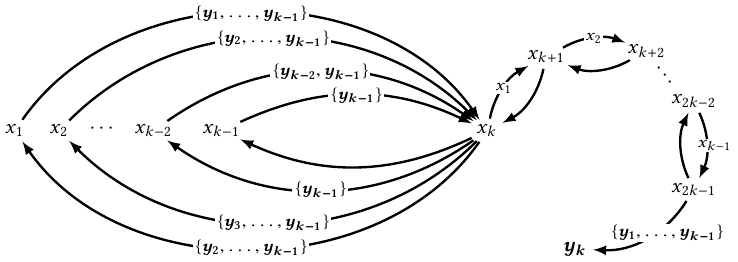} 
     \caption{The quantifier-free part of $\Phi$ from Example~\ref{example:unsat2}.}
     \label{fig:example_unsat2}
 \end{figure}   
\begin{example} \label{example:unsat2} 
 Consider the instance $\Phi\coloneqq \exists x_1\forall y_1\exists x_2\forall y_2\cdots \exists x_k\forall y_k\exists x_{k+1}\dots \exists x_{2k-1}\ \phi  $ with $\phi$ described by Figure~\ref{fig:example_unsat2}, where an edge from $x$ to $z$ labeled with $y$ stands for $\Michal(x,y,z)$.
 An edge from $x$ to $z$ labeled with some subset $A$ of the universal variables stands for a constraint of the form $\big((\bigwedge_{v\in A} x=v) \Rightarrow x\geq z\big)$ already derived by Algorithm~\ref{algo:dima}.
Using Lemma~\ref{lemma:pp-def}, these edges can be appropriately replaced with pp-definitions, and thus $\Phi$ is well-defined.

 We claim that the UP has the unique winning strategy on $\Phi$ of playing $\ass{y_i}$ equal to an arbitrary number $>\ass{x_i}$ if  $i=k$ and $\ass{x_1}=\cdots = \ass{x_k}$, and equal to $\min\{\ass{x_1},\dots, \ass{x_i}\}$ otherwise. We start by showing that this is a winning strategy.
 %
%

Suppose, on the contrary, that there exists an assignment $\ass{\cdot} \colon \vars \rightarrow\mathbb{Q}$ of values to the variables witnessing that the EP has a counter-strategy to the strategy of the UP from above.
First, consider the case where  $\ass{x_1},\dots,\ass{x_k}$ are not all equal.
Suppose that $\ass{x_k}=\min\{\ass{x_1},\dots,\ass{x_k}\}$ and let $j\in [k]$ be largest index such that $\ass{x_j}>\ass{x_k}$.
Recall that the algorithm already derived the constraint $\psi_1\coloneqq \big((\bigwedge\nolimits_{v\in \{y_{j+1},\dots,y_{k-1}\}} x_k=v) \Rightarrow x_k\geq x_j\big)$ on $\Phi$.
By the strategy of the UP, we have $\ass{y_{j+1}}=\cdots=\ass{y_{k-1}}=\ass{x_{j+1}}=\ass{x_k}.$
But then $\psi_1$ is clearly not satisfied by $\ass{\cdot}$, a contradiction.
Suppose now that $\ass{x_k}>\min\{\ass{x_1},\dots,\ass{x_k}\}$. 
Let $j\in [k]$ be the largest index such that $\ass{x_j}=\min\{\ass{x_1},\dots,\ass{x_k}\}$.
Recall that the algorithm already derived the constraint $\psi_2\coloneqq \big((\bigwedge\nolimits_{v\in \{y_j,\dots,y_{k-1}\}} x_j=v) \Rightarrow x_j\geq x_k\big)$. 
$\Phi$.
By the strategy of the UP, we have $\ass{y_j}=\cdots=\ass{y_{k-1}}=\ass{x_j}<\ass{x_k}.$
But then $\psi_2$ is clearly not satisfied by $\ass{\cdot}$, a contradiction.
We conclude that $\ass{x_1}=\cdots=\ass{x_k}$. 
In this case, the UP played $\ass{y_k}>\ass{x_k}$.
Since $\ass{\cdot}$ is a satisfying assignment, we must have $\ass{x_k}=\ass{x_{k+1}}=\cdots=\ass{x_{2k-1}}$.
But then $\ass{\cdot}$ does not satisfy $\big((\bigwedge\nolimits_{v\in \{y_{1},\dots,y_{k-1}\}} x_{2k-1}=v) \Rightarrow x_{2k-1}\geq y_k\big)$ because $\ass{y_k}>\ass{x_k}=\ass{x_{2k-1}}=\ass{y_1}=\dots=\ass{y_{k-1}}$, a contradiction.
We conclude that the strategy of the UP from above is a winning strategy.

The strategy of the UP is unique in the sense that, no matter what values the EP played for $x_1,\dots, x_{i}$, if the UP deviates from his strategy at $y_{i}$, then the EP wins by playing  $\ass{x_{i+1}}=\cdots =\ass{x_{2k-1}}$ equal to an arbitrary number $>\max( \ass{x_i}, \ass{y_i})$ if $\ass{x_i}\neq \ass{y_i}$ and equal to $\min\{\ass{x_1},\dots,\ass{x_i}\}$ otherwise.
%

Finally, we show that the algorithm derives \emph{false}.
In the first run of the main loop, we get  $(x_{i+1}\geq x_{i})$ for all $i\in \{1,\dots, k-2\}$.
Assuming previously derived constraints, we get $(x_{k}\geq x_{i})$ for all $i\in \{1,\dots, k-1\}$ (purely by transitivity).
Now it is possible to derive $(x_{i}=y_{i}\Rightarrow x_{i} \geq x_{i+1})$ for all $i\in \{1,\dots, k-1\}$.
In the final step, we get $(x_1\geq y_k)$, again, simply by invoking an oracle for $\CSP(\mathbb{Q};\Michal,<)$, which makes the algorithm reject.
\end{example}

\subsection{False instances}

 The goal of this subsection is reformulating the condition for rejection by Algorithm~\ref{algo:dima} within a certain proof system $\mathcal{P}$ operating on $\Phi$, whose rules are given in Table~\ref{table:proof_system} using a Datalog-style syntax. 
 The proof system syntactically derives predicates of the form $\mathcal{P}(x,y;A)$ with $x,y\in V$ and $A \subseteq \varsu$ (on the left hand side of $\mathrel{\mathop:-}$) from other predicates of this form derived earlier and the information encoded in $\Phi$ (on the right hand side of $\mathrel{\mathop:-}$).

We shall now provide some intuition behind the formulation of the proof system.  
The idea is that an expression  $\mathcal{P}(x,z;A)$ should capture a constraint $\big((\bigwedge\nolimits_{v\in A} x=v) \Rightarrow x\geq z \big)$  entailed by $\Phi$, where  $A$ only consists of universal variables.
Assuming said semantics, \Trans and \AltTrans   just describe natural properties of such expressions, and \Simplify and \Refute witness consequences of the quantification over the variables.
The combination of \Init and \Progress captures precisely the situations where the UP can indirectly enforce the identification of two (potentially existential) variables within a constraint in $\phi$.
In particular, it can be used to introduce $\mathcal{P}(u,z; \{ v\})$ for  conjuncts $(u = v \Rightarrow v \geq z)$ in $\phi$ with $v$ universally quantified as follows. The proof system first derives $\mathcal{P}(u,u; \emptyset)$ and $\mathcal{P}(v,v; \emptyset)$ using \Init. Then it uses 
\Progress to derive $\mathcal{P}(u,z; \{ v\})$ by identifying $x_1, x_2$ with $u$ and $x_3, x_4$ with $v$.   

The reader might naturally ask why we cannot obtain a polynomial-time algorithm by just closing $\Phi$ under the rules of the proof system with a suitable form of fixed-point semantics. The reason is that, already under the least fixed-point semantics, the proof system might derive exponentially many expressions of the form $\mathcal{P}(x,z;A)$. Such a situation occurs, e.g., in Example~\ref{example:five} and in the case of the constraint paths in $\phi$ as defined in the proof of Lemma~\ref{lemma:hardness}.

The precise connection between the proof system and Algorithm~\ref{algo:dima} is captured by Lemma~\ref{lemma:algo_p_halt}.  
 \begin{table} 
    \caption{The inference rules of the proof system $\mathcal{P}$. Here I stands for “initialize”, S for “simplify”, T for “transitivity”, A for “alternative transitivity”, C for “constraint,” and R for “reject.”} 
    \label{table:proof_system}    
     \thinmuskip=1mu
     \medmuskip=1mu
    \thickmuskip=1mu
   \begin{tabular}{ll} 
   \toprule  
         \Init  \label{rule:Init} & $\mathcal{P}(x,x;\emptyset) \  \mathrel{\mathop:-} \ x\in V$     \\ \midrule 
         \Simplify  \label{rule:Simplify} & $\mathcal{P}(x,z;A\setminus \cut{x}{z}) \  \mathrel{\mathop:-} \ \mathcal{P}(x,z;A)$   \\ \midrule
        \Trans  \label{rule:Trans}  & $ \mathcal{P}(x,z;A) \  \mathrel{\mathop:-} \ \mathcal{P}(x,y;A) \wedge \mathcal{P}(y,z;\emptyset)$    \\ \midrule 
            \Refute \label{rule:Refute} &   
                $ \bot \  \mathrel{\mathop:-} \     \left\{\hspace{-0.5em}\begin{array}{l} \text{1. $\mathcal{P}(x,z;\emptyset)$}  \\
           \text{2. $x\prec z$ and $z\in \varsu$, or $z\prec x$ and $x\in \varsu$} \end{array}\right.$  \\  \midrule 
     \AltTrans  \label{rule:AltTrans} 
      & $\mathcal{P}(x_i,z;A\cup B\cup (\{x_1,x_2\} \setminus \{x_i\})) \  \mathrel{\mathop:-} \ \left\{\hspace{-0.5em}\begin{array}{l} \text{1. $\mathcal{P}(x_1,y;A) \wedge \mathcal{P}(y,x_2;\emptyset) \wedge \mathcal{P}(y,z;B)$} \\
           \text{2. $(\{x_1,x_2\} \setminus \{x_i\})\subseteq \varsu$  $(i\in \{1,2\}$)} \end{array}\right.$     \\ 
       \midrule
      \Progress  \label{rule:Progress} & 

 $ \mathcal{P}(x_i,z;A\cup B\cup  (\{x_1,x_2,x_3,x_4\} \setminus \{x_i\}) ) \  \mathrel{\mathop:-} \     \left\{\hspace{-0.5em}\begin{array}{l} \text{1.  $\mathcal{P}(x_1,u;A)\wedge   \mathcal{P}(u,x_2;\emptyset)$}  \\
           \text{2.  $\mathcal{P}(x_3,v;B)  \wedge   \mathcal{P}(v,x_4;\emptyset)$} \\
           \text{3.  $(\{x_1,x_2,x_3,x_4\} \setminus \{x_i\})\subseteq \varsu$  $(i\in \{1,2,3,4\}$)} \\
           \text{4.    $(u=v \Rightarrow u \geq z) $ or $ (v=u \Rightarrow v \geq z)$ in $\phi$}
           \end{array}\right.$           \\ 
           \bottomrule
    \end{tabular}
\end{table}  
Note that Lemma~\ref{lemma:algo_p_halt} in particular implies that Algorithm~\ref{algo:dima} rejects whenever the proof system derives $\bot$.
When combined with Lemma~\ref{lemma:algo_refutes_correctly} and Lemma~\ref{lemma:winning_EP} (proved later in Section~\ref{section:satisfiable}), we get that this is in fact the only situation in which Algorithm~\ref{algo:dima} rejects.   
Lemma~\ref{lemma:algo_p_halt} can be proved by a straightforward induction on derivation sequences within $\mathcal{P}$. 
%
%
\begin{restatable}{lemma}{AlgoHalt}  \label{lemma:algo_p_halt} Suppose that $\mathcal{P}(x,z;A)$ is derived by the proof system and $z\notin A$. Then Algorithm~\ref{algo:dima} expands $\phi$ by 
$ 
    \big( (\bigwedge\nolimits_{v\in \uparrow_A\setminus (\{x,z\}\cup \cut{x}{z})} x=v) \Rightarrow x\geq z \big).  
$
 In particular, it expands $\phi$ by $(x\geq z)$ for every derived $\mathcal{P}(x,z;\emptyset)$.
\end{restatable}

We conclude this subsection with two examples, the first showcasing how the run of Algorithm~\ref{algo:dima} can be represented within the proof system, and the second showcasing that, in general, the proof system cannot be used to verify true instances in polynomial time.  
\begin{example} \label{example:algo-comp2} Consider the instance $\Phi$ from Example~\ref{example:algo-comp}.
We show that the proof system derives  $\bot$.
First, we can derive $\mathcal{P}(x_i,x_i;\emptyset)$ for every $i\in [5]$ using \Init.
With \Progress (and suitable identifications of variables), we get $\mathcal{P}(x_3,x_1;\emptyset), \mathcal{P}(x_5,x_1;\emptyset)$, $\mathcal{P}(x_1,x_5;\{x_2\})$, $\mathcal{P}(x_3,x_4;\{x_2\})$, and $\mathcal{P}(x_5,x_3;\{x_4\})$.
Next, a single application of \AltTrans yields $\mathcal{P}(x_1,x_3;\{x_2,x_4\})$.
We can use \Simplify to simplify the latter to $\mathcal{P}(x_1,x_3;\{x_2\})$.
Using \AltTrans again, we get $\mathcal{P}(x_1,x_4;\{x_2\})$, and finally, \Simplify simplifies the latter to $\mathcal{P}(x_1,x_4;\emptyset)$. Now an application of  \Refute yields $\bot$. 
\end{example}  
 \begin{figure}[h]
     \centering
      \includegraphics[width=0.6\linewidth]{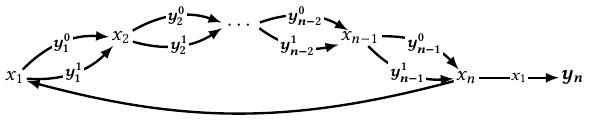} 
     \caption{The quantifier-free part of $\Phi$ from Example~\ref{example:five}.}
     \label{fig:example_5}
 \end{figure} 
\begin{example}  \label{example:five}
 Consider $\Phi\coloneqq \exists x_1\forall y_1^0\forall y_1^1\exists x_2\forall y_2^0\forall y_2^1 \cdots \exists x_{n-1}\forall y_{n-1}^0 \forall y_{n-1}^1\exists x_n\forall y_n \ \phi$ with $\phi$ described by Figure~\ref{fig:example_5},  where an edge from $x$ to $z$ labeled with $y$ stands for $\Michal(x,y,z)$.
Note that the proof system derives $\mathcal{P}(x_1,x_n;\{y^{i_1}_{1},\dots,y^{i_{n-1}}_{n-1}\})$ for all $i_1,\dots, i_{n-1}\in \{0,1\}$.
Indeed, this is because it can follow the shortest derivation sequences, of which there are exponentially many.
In contrast, Algorithm~\ref{algo:dima} derives
the constraints
$(x_{n-1}\geq x_1),\dots, (x_{2}\geq x_1)$, $(x_1 \geq y_n)$ in this order, which leads to rejection.
Interestingly enough, constraint paths as in Figure~\ref{fig:example_5} were previously used in \cite{zhuk2023complete} to prove PSPACE-hardness of $\QCSP(\mathbb{Q};\Dima)$.
\end{example}

\subsection{True instances\label{section:satisfiable}} 
In this subsection we prove Lemma~\ref{lemma:winning_EP}, which states that the refutation condition \Refute from Table~\ref{table:proof_system} is not only sufficient but also necessary.
%
%
\begin{lemma} \label{lemma:winning_EP} If the proof system does not derive $\bot$ from $\Phi$, then $\Phi$ is true. 
\end{lemma}
\begin{proof} Suppose that the proof system cannot derive $\bot$ from $\Phi$. Consider the following strategy for the EP. 
Let $x\in \varse$ be such that $\ass{x}$ is not yet defined but $\ass{z}$ is defined for every $z\prec x$.
Then the EP selects any value for $x$ 
such that, for every  $z\prec x$:
\begin{itemize}
    \item $\ass{x}\geq \ass{z}$ \emph{if and only if} there exists $y\prec x$ with $\ass{y}\geq \ass{z}$ and $\mathcal{P}(x,y;\emptyset)$;  
    \item $\ass{x}=\ass{z}$ \emph{if and only if} there exist $y_1,y_2\prec x$ and $y_{2}\prec A \prec x$ such that \[\mathcal{P}(x,y_1;\emptyset)\wedge \mathcal{P}(y_2,x;A) \quad\text{and}\quad    \ass{z}=\ass{y_1}=\ass{y_2}=\ass{A}.\]
\end{itemize}

\begin{claim} \label{claim:strategy_EP_well_defined} The strategy of the EP is well-defined. 
\end{claim}
   
\begin{claimproof} Suppose, on the contrary, that it is not. 
Let $x\in \varse$ be the smallest variable w.r.t.\ $\prec$ for which the strategy of the EP is not well-defined.
Then it must be the case that there exist $y,y_1,y_2\prec x$ and $y_{2}\prec A \prec x$  such that  
\begin{equation} 
    \begin{array}{rl}
           \mathcal{P}(x,y;\emptyset) \wedge  \mathcal{P}(x,y_1;\emptyset)\wedge \mathcal{P}(y_2,x;A)  \quad \text{and} \quad 
         \ass{y}>\ass{y_1}=\ass{y_2}=\ass{A}.
    \end{array}
    \label{eq:EP_strat_not_defined}
\end{equation} 
In particular, $y \not\in A$.
We choose the smallest possible $y$ w.r.t.\ $\prec$ witnessing a condition of the form \eqref{eq:EP_strat_not_defined}.
By \Trans , we have $\mathcal{P}(y_2,y;A)$.

 \case{1}{$y\prec y_2$}. Then, by \Simplify , we have $\mathcal{P}(y_2,y;\emptyset)$.  

\case{1.1}{$y_{2}\in \varsu$}. Then  $\bot$ can be derived using \Refute, a contradiction.

\case{1.2}{$y_{2}\in \varse$}. Then the EP did not follow his strategy because $\ass{y}>\ass{y_2}$ and we have $\mathcal{P}(y_2,y;\emptyset)$, a contradiction.  

\case{2}{$y_2\prec y$}.

\case{2.1}{$y\in \varsu$}.
Then, by \Simplify , we have $\mathcal{P}(y_2,y;\emptyset)$. But then $\bot$ can be derived using \Refute, a contradiction.

\case{2.2}{$y\in \varse$}. 
Then, by \Simplify , we have $\mathcal{P}(y_2,y;A\setminus \cut{y_2}{y})$. 
Since $\ass{y}\geq \ass{y_2}$, by the strategy of the EP, there  exists a variable $y'\prec y$ such that $\ass{y'}\geq \ass{y_2}$ and $\mathcal{P}(y,y';\emptyset)$.
If $\ass{y'}=\ass{y_2}$, then the EP did not follow his strategy, because he played $\ass{y}>\ass{y'}$ while $y'\prec y$, $y_2\prec A\setminus \cut{y_2}{y} \prec y$, $\mathcal{P}(y,y';\emptyset) \wedge \mathcal{P}(y_2,y;A\setminus \cut{y_2}{y})$, and $\ass{y'}=\ass{y_2}=\ass{A\setminus \cut{y_2}{y}}$. 
a contradiction. So it must be the case that $\ass{y'}>\ass{y_2}$.
By \Trans , we have $\mathcal{P}(y_2,y';A)$.
But now $y'$ can assume the role of $y$ in~\eqref{eq:EP_strat_not_defined}, a contradiction to the minimality of $y$ w.r.t.\ $\prec$.  
\end{claimproof} 
 The next claim characterizes the equality of values for pairs of variables under $\ass{\cdot}$ in terms of properties of previously quantified variables, assuming that the EP has followed the strategy above.
In particular, we show that if $\ass{x}=\ass{z}$ if and only if there exists a variable $y\preceq \{x,z\}$ so that $\ass{x}=\ass{y}$ and $\ass{y}=\ass{z}$ are enforced by the identifications of values of universal variables with $y$ by the UP.
Recall the comparison relations $\preceq$ and $\preceq_{\forall}$ from Definition~\ref{def:terminology}.  
\begin{restatable}{claim}{CharactGame}  \label{claim:charact_game}
     Suppose that the EP follows  
the strategy above. Then, for all $x,z\in \vars$, we have $\ass{x}=\ass{z}$  if and only if there exist $x_1,x_2,z_1,z_2\in V$ and   $A_{x_2,x},A_{z_2,z}\subseteq \varsu$
such that 
\begin{enumerate} 
    \item \label{item:ugly_claim_1} $\{x_1,x_2\}\preceq x$ and $\{z_1,z_2\}\preceq z$
    \item \label{item:ugly_claim_2}  $x_2\prec A_{x_2,x}\preceq x$ and  $z_2 \prec A_{z_2,z} \preceq z$, 
    \item \label{item:ugly_claim_3}  $y\preceq_{\forall} \{x_1,x_2,z_1,z_2\}$ for some $y\in \{x_2,z_2\}$,
    \item \label{item:ugly_claim_4}  $\mathcal{P}(x_2,x;A_{x_2,x})\wedge \mathcal{P}(x,x_1;\emptyset)\wedge \mathcal{P}(z_2,z;A_{z_2,z})\wedge \mathcal{P}(z,z_1;\emptyset) $, 
    \item \label{item:ugly_claim_5}  $\ass{A_{x_2,x}}=\ass{x_1}=\ass{x_2} =\ass{z_1}=\ass{z_2}= \ass{A_{z_2,z}}$.  
\end{enumerate}
Whenever the right-hand side of the equivalence holds, we also have $\ass{x}=\ass{x_2}$ and $  \ass{z}=\ass{z_2}.$
\end{restatable}  
\begin{claimproof}[Proof idea]  ``$\Leftarrow$''
We show that $\ass{x}=\ass{x_2}$ and $\ass{z}=\ass{z_2}$.
If $x\equiv x_2$, then clearly $\ass{x}=\ass{x_2}$. So, w.l.o.g., $x_2\prec x$.  
If $x\in \varsu$, then \Simplify yields $\mathcal{P}(x_2,x;\emptyset)$ and hence \Refute produces $\bot$, a contradiction. 
So we must have $x\in \varse$. Then either $x_1\equiv  x$ or $x_1 \prec x$, and it follows from the strategy of the EP that $\ass{x}=\ass{x_2}$.
Analogously we obtain that
$\ass{z}=\ass{z_2}$.
The rest follows by the transitivity of the equality.

``$\Rightarrow$'' Whenever the  right-hand side of the equivalence in Claim~\ref{claim:charact_game} is satisfied, we  call   $(x,x_1,x_2;A_{x_2,x})$ and $(z,z_1,z_2;A_{z_2,z})$ \emph{witnessing quadruples} for $\ass{x}=\ass{z}$.
If $x\equiv z$, then the statement trivially follows using \Init ,  the witnessing quadruples are $(x,x,x;\emptyset)$ and $(z,z,z;\emptyset)$.
So, w.l.o.g., $z\prec x$.
If $x\in \varsu$, then the claim follows using \Init , the witnessing quadruples are again $(x,x,x;\emptyset)$ and $(z,z,z;\emptyset)$.  
So suppose that $x\in \varse$ and  that the claim holds for all pairs of variables preceding $x$.
Since $\ass{x}=\ass{z}$, by the strategy of the EP, there exist $x_1,x_2\prec x$ and $x_{2}\prec A \prec x$ such that $\mathcal{P}(x,x_1;\emptyset)\wedge \mathcal{P}(x_2,x;A)$ and   $\ass{z}=\ass{x_1}=\ass{x_2}=\ass{A}.$ 

Since $\ass{x_2}=\ass{z}$ and $x_2,z\prec x$, we can apply the induction hypothesis for the pair $x_2,z$ to obtain the witnessing quadruples $(x_2,x_{2_1},x_{2_2};A_{x_{2_2},x_2})$ and $(z,z_{1},z_{2};A_{z_{2},z})$.
By assumption, there exists $y\in \{z_2,x_{2_2}\}$ such that $y\preceq_{\forall} \{z_1,z_2,x_{2_1},x_{2_2}\}$.
Note that $\ass{y}=\ass{x_1}$.
Thus, we can apply the induction hypothesis for the pair $x_1,y$ to obtain the witnessing quadruples $(x_1,x_{1_1},x_{1_2};A_{x_{1_2},x_1})$ and $(y,y_{1},y_{2};A_{y_{2},y})$.  
By assumption, there exists $y'\in \{y_2,x_{1_2}\}$  such that $y'\preceq_{\forall} \{y_1,y_2,x_{1_1},x_{1_2}\}$.  
The two cases $y \equiv z_2$ and $y \equiv x_{2_2}$ are illustrated in Figure~\ref{figure:cases}.

\begin{figure}[h]
 \begin{center}   \includegraphics[width=0.49\linewidth]{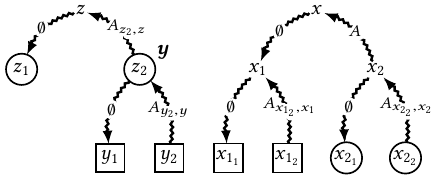}  \includegraphics[width=0.49\linewidth]{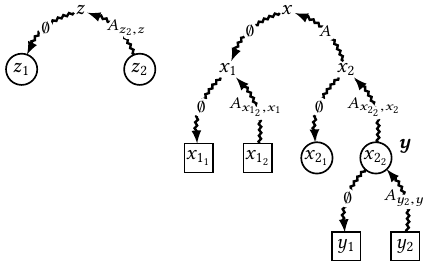}    
\end{center} 
   \caption{\label{figure:cases}Cases $y \equiv z_2$ and $y \equiv x_{2_2}$ in the proof of Claim~\ref{claim:charact_game}. The squiggly arrows represent inferences of $\mathcal{P}$.}
\end{figure}

  Our goal is to find witnesses $x'_1,x'_2,z'_1,z'_2$ for the main statement of the claim, i.e., the witnessing quadruples will be of the form $(x,x'_1,x'_2;A_{x'_2,x})$ and $(z,z'_1,z'_2;A_{z'_2,z})$. 
The idea is that we want to choose $x_1'$, $x_2'$, $z_1'$, $z_2'$ from the variables introduced above, which all evaluate to the value $\ass{x}=\ass{z}$ in $\ass{\cdot}$. To obtain the property in item~3, we want to choose variables that are small enough with respect to the order $\prec$, so that one of them can be compared to the others with respect to $\preceq_{\forall}$, assuming the properties of $y$ and $y'$ from above.

One suitable choice of witnesses is as follows. 
First, we choose $x_1'\coloneqq x_{1_1}$. As visible in Figure~\ref{figure:cases}, we can apply \Trans to $\mathcal{P}(x,x_1;\emptyset)\wedge\mathcal{P}(x_1,x_{1_1};\emptyset)$ to derive $\mathcal{P}(x,x_{1_1};\emptyset)$. 
Second, we choose $z'_1\coloneqq z_1$ if $y \not\equiv z_1$ and  $z'_1\coloneqq y_1$ otherwise. 
Note that we have $\mathcal{P}(z,z_1;\emptyset)$ by assumption and, if $y\equiv z_1$, then we can use  \Trans to derive $\mathcal{P}(z,y_1;\emptyset)$ from $\mathcal{P}(z,z_1;\emptyset)\wedge\mathcal{P}(y,y_1;\emptyset)$.  
Next, we choose $x'_2\coloneqq x_{2_2}$ if $y \not\equiv x_{2_2}$, and $x'_2\coloneqq y_2$ otherwise. 
A short argument shows that $x_{2_2}\preceq_{\forall} x_{2_1}$, which allows us to apply \AltTrans to $\mathcal{P}(x_{2_2},x_2;A_{x_{2_2},x_2})\wedge \mathcal{P}(x_{2},x_{2_1};\emptyset)\wedge \mathcal{P}(x_2,x;A)$ to obtain an expression of the form
$\mathcal{P}(x_{2_2},x;A_{x_{2_2},x})$.
If $y  \equiv x_{2_2}$, then it is necessary to apply \AltTrans a second time to obtain the expression of the form $\mathcal{P}(y_2,x;A_{y_2,x})$.
Finally, the choice for $z'_2$ that we need will be $z'_2 \coloneqq z_2$ if $y \not\equiv z_2$ or $z'_2 \coloneqq y_2$ otherwise. 

With the above witnessing quadruples, one can verify that items 1, 2, 4, and 5 will be satisfied.
Thanks to choosing ``small enough candidates'' with respect to $\prec$ for each of $x_1'$, $x_2'$, $z_1'$, $z_2'$, item 3 can be verified as well. A full proof of Claim~\ref{claim:charact_game} with a verification of these properties can be found in Appendix~\ref{appendix:proofs}. 
 \end{claimproof}
 
\begin{claim} \label{claim:winning_ep} The strategy of the EP is a winning strategy.
\end{claim}

\begin{claimproof}
Suppose, on the contrary, that this is not the case.
Then there has to be a violated constraint of the form $(x=z \Rightarrow x\geq w)$, i.e., $\ass{x}=\ass{z}<\ass{w}$.
Since $\ass{x}=\ass{z}$, by Claim~\ref{claim:charact_game}, there exist $x_1,x_2,z_1,z_2\in V$ and   $A_{x_2,x},A_{z_2,z}\subseteq \varsu$ such that 
\begin{itemize}
    \item $\{x_1,x_2\}\preceq x$ and $\{z_1,z_2\}\preceq z$
    \item $x_2\prec A_{x_2,x}\preceq x$ and  $z_2 \prec A_{z_2,z}\preceq z$, 
    \item $y\preceq_{\forall} \{x_1,x_2,z_1,z_2\}$ for some $y\in \{x_2,z_2\}$,
    \item $\mathcal{P}(x_2,x;A_{x_2,x})\wedge \mathcal{P}(x,x_1;\emptyset)\wedge \mathcal{P}(z_2,z;A_{z_2,z})\wedge \mathcal{P}(z,z_1;\emptyset) $, 
    \item $\ass{A_{x_2,x}}=\ass{x_1}=\ass{x_2} =\ass{z_1}=\ass{z_2}= \ass{A_{z_2,z}}$.
\end{itemize}
Moreover, 
we have $\ass{x}=\ass{y}=\ass{z}$.
Let  $A\coloneqq A_{x_2,x}\cup A_{z_2,z}\cup(\{x_1,x_2,z_1,z_2\}\setminus \{y\}).$ 
By a single application of \Progress , we get $\mathcal{P}(y,w;A)$.
Since $\ass{z}<\ass{w}$, clearly $w$ is different from all variables which share the value with $z$. 
We choose the smallest possible $w$ w.r.t.\ $\prec$ for which $\mathcal{P}(y,w;A)$ can be derived and such that $\ass{z}<\ass{w}$.
Since $\ass{w}\neq \ass{z}$, we have $w\notin A$.
Now we consider the following cases.

\case{1}{$w\prec y$}. By \Simplify , we have  $\mathcal{P}(y,w;\emptyset)$. 
 
\case{1.1}{$y\in \varsu$}. In this case $\bot$ can be derived using \Refute, a contradiction.

\case{1.2}{$y\in \varse$}. Then
 the EP  was supposed to set $\ass{y} \geq \ass{w}$, however, we have $ \ass{y} < \ass{w}$. Hence, the EP did not follow his strategy, a contradiction. 
 
 \case{2}{$y\prec w$}. 
 
\case{2.1}{$w\in \varsu$}. 
 By \Simplify, we get $\mathcal{P}(y,w;\emptyset)$. But then $\bot$ can be derived using \Refute, a contradiction.

\case{2.2}{$w\in \varse$}. Since $\ass{w}>\ass{y}$, by the strategy of the EP, there must exist $w'\prec w$ with $\ass{w'}\geq \ass{y}$ such that $\mathcal{P}(w,w';\emptyset)$.
 By \Simplify , we have $\mathcal{P}(y,w; A\setminus \cut{y}{w}).$
 If $\ass{w'}=\ass{y}$, then the EP did not follow his strategy because he played $\ass{w}>\ass{w'}$ while $y,w'\prec w$ and $y  \prec A\setminus \cut{y}{w} \prec w$,
$\ass{w'}=\ass{y}=\ass{A\setminus \cut{y}{w}}$,  and $\mathcal{P}(y,w;A\setminus \cut{y}{w})  \wedge   \mathcal{P}(w,w';\emptyset), 
$
a contradiction. So it must be the case that $\ass{w'}>\ass{y}$. 
 By an application of \Trans  to $\mathcal{P}(y,w;A)\wedge\mathcal{P}(w,w';\emptyset)$, we have $\mathcal{P}(y,w';A).$
But note that now $w'$ can assume the role of $w$, a contradiction to the minimality of $w$ w.r.t.\ $\prec$.
\end{claimproof}  
 This concludes the proof of Lemma~\ref{lemma:winning_EP}.\end{proof}

\begin{lemma}\label{lemma:algo-poly}
Algorithm~\ref{algo:dima} solves $\QCSP(\mathbb{Q};\Michal)$ in polynomial time.
\end{lemma}

\begin{proof} Observe that Algorithm~\ref{algo:dima} runs in polynomial time with respect to the length of $\Phi$.
    Indeed, it expands $\Phi$ by at most $V^3$-many constraints, all of which have pp-definitions
    in $(\mathbb{Q};\Michal,<)$ of linear length due to Lemma~\ref{lemma:pp-def}, and $\CSP(\mathbb{Q};\Michal,<)$ is solvable in polynomial time~\cite{BoRyTCSPs}.  
    Note that, if $\Phi$ contains a conjunct $(x \geq z)$ or $(z \geq x)$ such that $x \prec z$ and $z \in \varsu$, then $\Phi$ is false. 
    Therefore, by Lemma~\ref{lemma:algo_refutes_correctly}, $\Phi$ is false whenever the algorithm rejects. Suppose that the algorithm accepts an instance $\Phi$. By Lemma~\ref{lemma:algo_p_halt}, $\bot$ cannot be derived from $\Phi$ using the proof system and hence, by Lemma~\ref{lemma:winning_EP}, $\Phi$ is true. This completes the proof.
\end{proof}
  
 \begin{proof}[Proof of Theorem~\ref{thm_ptime}]  By Proposition~\ref{ppsynt},  every relation has a CNF-definition $\phi$ over the signature $\{\neq,\geq\}$ where each conjunct is of the form \eqref{eq:CNF} 
  for $k,\ell \geq 0$.
  By Lemma~\ref{lemma:ElimMin}, we can choose $\phi$ so that every conjunct of the form \eqref{eq:CNF} in $\phi$ satisfies $\ell\leq 1$. 
  It follows from Lemma~\ref{lemma:pp-def} by induction that every formula of the form \eqref{eq:CNF} with $\ell=1$ is pp-definable in $(\mathbb{Q};\Michal)$.
  The case with $\ell=0$ is qpp-definable in $(\mathbb{Q};\Michal)$ by universally quantifying over $z_1$ in the case $\ell=1$.
  By Proposition~\ref{prop:reductions}, $\QCSP(\struct{B})$ reduces in LOGSPACE to $\QCSP(\mathbb{Q};\Michal)$.
  By Lemma~\ref{lemma:algo-poly}, $\QCSP(\mathbb{Q};\Michal)$ is solvable in polynomial time.
  \end{proof}

\section{Identifying the hard cases\label{section:second_proof}}

This section is devoted to the proof of Theorem~\ref{thm_hardness}.
%

First, we use a syntactical argument to reduce the arity of the relations that need to considered to $4$.
%
%
\begin{restatable}{lemma}{ArityFour} 
  \label{lemma:arity_four} Let $\struct{B}$ be an OH structure that is not preserved by $\pp$. Then $\struct{B}$ pp-defines a relation of arity at most $4$ that is not preserved by $\pp$.
\end{restatable}

Second, we perform an ``educated brute-force'' search through all relations of arity at most $4$ that are not preserved by $\pp$ in order to classify them.
Recall the relations $\Dima$ and  $\Zaneta$ defined in the introduction.
%
%
\begin{restatable}{lemma}{NotPP} 
\label{lem:notpp}
     Let $\struct{B}$ be an OH structure that is not preserved by $\pp$. Then $\struct{B}$ pp-defines $\Dima$ 
     or $(\struct{B};\Michal)$ qpp-defines  $\Zaneta$.
\end{restatable}

Third, we show coNP-hardness of the QCSP for said relations combined with  $\Michal$.
\begin{lemma} \label{lemma:hardness} Let $\struct{B}$ be an OH structure that is not preserved by $\pp$ and pp-defines $\Michal$. Then $\QCSP(\struct{B})$ is coNP-hard. 
\end{lemma}
The proof of Lemma~\ref{lemma:hardness} below relies almost entirely on constraint paths built using $\Michal$.
In Figure~\ref{fig:dima_figure_Ci}, edges relate to constraints over $\Michal$, e.g. the two leftmost arrows in the lower chain represent $\Michal(f,y_{1}^i, z)\wedge \Michal(z,z,f)$ for $i \in \{1,2\}$ and $z$ corresponds to an unlabelled vertex.
These constraint paths are used to generate exponentially many incomparable expressions within the proof system $\mathcal{P}$, but $\Michal$ itself has no mechanism for turning them into a working gadget.
This is why such constraint paths can be handled by Algorithm~\ref{algo:dima}.
The situation changes already when we add a single constraint associated to the relation $\Zaneta$.  
\begin{figure}[ht]
 \begin{center}   \includegraphics[width=0.55\linewidth]{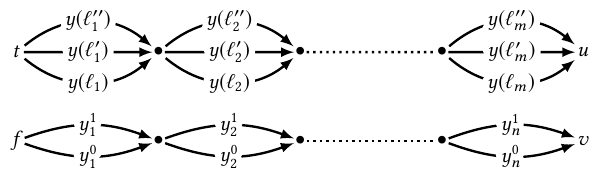} 
\end{center}     \caption{A gadget for the proof of Lemma~\ref{lemma:hardness}.}
    \label{fig:dima_figure_Ci}
\end{figure}
\begin{proof}[Proof of Lemma~\ref{lemma:hardness}]
In the case where $\struct{B}$ pp-defines $\Dima$, we have that
 $\QCSP(\struct{B})$ is PSPACE-hard by Corollary~6 in \cite{zhuk2023complete} and Proposition~\ref{prop:reductions}.
So suppose that $\struct{B}$ does not pp-define $\Dima$.
By Lemma~\ref{lem:notpp}, we have that $\struct{B}$ qpp-defines $\Zaneta$.

We reduce from the complement of the satisfiability problem for propositional $3$-CNF.
Consider an arbitrary propositional $3$-CNF formula $\psi$, i.e., a conjunction of clauses of the form $\ell_i\vee {\ell}'_i\vee {\ell}''_i$ for $i\in [m]$, where $\ell_i,{\ell}'_i,{\ell}''_i$ are potentially negated propositional variables from $\{x_1,\dots, x_n\}$.
We set $\Phi\coloneqq  \exists t\exists f \forall y^0_1 \forall y^1_1 \dots \forall y^0_n \forall y^1_n \exists \dots \exists u\exists v \ \phi \wedge \Zaneta(v,f,u,t),$
where $\exists \dots$ are additional unlabelled existentially quantified variables, and $\phi$ is defined as in Figure~\ref{fig:dima_figure_Ci}.
Here $y(x_i)\coloneqq y_i^1$, $y(\neg x_i)\coloneqq y_i^0$, a directed edge from $x$ to $z$ labeled with $y$ stands for $\Michal(x,y,z)\wedge \Michal(z,z,x)$,\footnote{Note the difference from the previous interpretation of labeled directed edges, e.g., in Examples~\ref{example:unsat2} and~\ref{example:five}. The current interpretation entails $D(x,y,z)$.}   and unlabelled dots correspond to unlabelled existential variables.

 ``$\Rightarrow$''  Suppose that $\psi$ is not satisfiable. We show that the EP has a winning strategy on $\Phi$. First, the EP chooses $\ass{f}<\ass{t}$. 
If there exists $i \in [n]$, such that the UP chose $\ass{y^0_i}\neq \ass{f}$ and $\ass{y^1_i} \neq \ass{f}$, then the EP chooses the values for the remaining existential variables as follows: equal $\ass{f}$ if they appear in the lower chain in Figure~\ref{fig:dima_figure_Ci} before $y^0_i$ and $y^1_i$, and equal $\ass{t}$ otherwise. Since $\ass{f}<\ass{t}$, this choice satisfies $\phi \wedge (v \neq f)$. We may therefore assume that $\ass{f}\in \{ \ass{y^0_i},\ass{y^1_i}\}$ for every $i\in [n]$.

We claim that there exists $j \in [m]$ such that 
$ \ass{t} \notin \{\ass{y(\ell_j)}, \ass{y(\ell_j')}, \ass{y(\ell_j'')}\}. $
Suppose, on the contrary, that this is not the case. 
Let $\ass{\cdot}'$ be any map from $\{x_1, \dots, x_n\}$ to $\{0,1\}$ such that, for every $i\in [n]$,
$ \ass{x_i}' = 0 $ if $\ass{y_i^0}=t$ and $\ass{x_i}' = 1$ if $\ass{y_i^1}=t$. Recall that $\ass{f}\in \{ \ass{y^0_i},\ass{y^1_i}\}$ for every $i\in [n]$ and thus $\ass{\cdot}'$ is well-defined. Observe that $\ass{\cdot}'$ is a satisfying assignment to $\psi$, contradicting our assumption. Hence the claim holds.

 The EP can choose the values for the remaining existential variables as follows: equal $\ass{t}$ if they appear in the upper chain in Figure~\ref{fig:dima_figure_Ci} before the $j$-th column, equal to an arbitrary number $q>\ass{t}$ if they appear in the upper chain after the $j$-th column, and equal $\ass{f}$ otherwise. Such assignment satisfies $\phi \wedge (t \neq u)$.

``$\Leftarrow$'' Suppose that there exists a satisfying assignment $\ass{\cdot}'$ to $\psi$. We show that then the UP has a winning strategy on $\Phi$. If the EP chooses $\ass{f} \geq \ass{t}$, the UP wins on $\Phi$, suppose therefore that $\ass{f}<\ass{t}$. If $\ass{x_i}'=0$, the UP plays $\ass{y_i^0}=t$ and $\ass{y_i^1}=f$, and if $\ass{x_i}'=1$, the UP plays $\ass{y_i^0}=f$ and $\ass{y_i^1}=t$. It follows from the lower chain in Figure~\ref{fig:dima_figure_Ci} that the EP loses unless $\ass{v}=\ass{f}$. Moreover, since $\ass{\cdot}'$ is a satisfying assignment to $\psi$, it follows from the upper chain that the EP loses unless $\ass{u}=\ass{t}$. But then $\ass{\cdot}$ violates $(v \neq f \vee u \neq t)$ and the UP wins again.
\end{proof}
 \begin{proof}[Proof of Theorem~\ref{thm_hardness}]  
Let $\struct{B}$ be a structure which is not GOH, and preserved by neither $\pp$ nor $\dual\pp$.
 Then it follows from  Theorem~\ref{thm:noGOH} that $\QCSP(\struct{B})$ is coNP-hard, in which case we are done, or $\struct{B}$ pp-defines $\Michal$ or $\DMichal$.
 If $\struct{B}$ pp-defines $\Michal$, then $\QCSP(\struct{B})$ is coNP-hard by Lemma~\ref{lemma:hardness} because $\struct{B}$ is not preserved by $\pp$.
 Otherwise $\struct{B}$ pp-defines $\DMichal$. Then we reach the same conclusion as in the previous case using the dual version of each result in Section~\ref{section:second_proof}.
 These can be obtained simply by reversing the order in each individual statement.  
  \end{proof}

\section{Open questions}

 For quantified OH constraints, we leave the following question open:


\medskip\noindent \emph{Question 1:} Do OH QCSPs exhibit a dichotomy between coNP and PSPACE-hardness?

\medskip We also ask the following questions regarding open cases 
outside of OH:

\medskip\noindent  \emph{Question 2:} Is $\QCSP(\struct{B})$ in P whenever $\struct{B}$ is a temporal structure preserved by $\mi$~\cite{bodirsky2010complexity}?  It is enough to consider $\QCSP(\mathbb{Q};x_1\neq x_2\vee x_1\geq x_3 \vee x_1 > x_4)$.

\medskip\noindent \emph{Question 3:} Is $\QCSP(\struct{B})$ in NP whenever $\struct{B}$ is a temporal structure preserved by $\pp$? It is enough to consider  $\QCSP(\mathbb{Q};x_1\neq x_2\vee x_1\geq x_3 \vee x_1 \geq x_4)$.



\bibliography{local}

\appendix

\section{Omitted proofs\label{appendix:proofs}}

We use the bar notation for tuples; for a tuple $\bar{t}$ indexed by a set $I$, the value of $\bar{t}$ at the position $i\in I$ is denoted by  $\bar{t}\of{i}$. 
A tuple is called \emph{injective} if all of its entries are pairwise distinct.  
For a function $f\colon A^n \rightarrow B$ ($n\geq 1$) and $k$-tuples $\bar{t}_1,\ldots,\bar{t}_n \in A^{k}$,   
we use $f(\bar{t}_1,\ldots,\bar{t}_n)$ as a shortcut for the $k$-tuple 
$(f(\bar{t}_1\of{1},\ldots\bar{t}_n\of{1}),\dots, f(\bar{t}_1\of{k},\ldots,\bar{t}_n\of{k}) ).$

 The \emph{projection} of a relation $R$ of arity $n$ to a particular set of entries $I\subseteq [n]$ is the $|I|$-ary relation defined by the formula $\exists_{i\in [n]\setminus I}x_i \ R(x_1,\dots,x_{n})$.
When we work with tuples $\bar{t}$ in a relation defined by a formula $\phi(x_1,\dots,x_n)$,  we sometimes refer to the entries of $\bar{t}$ using the free variables of $\phi$, and write $\bar{t}\of{x_i}$ instead of $\bar{t}\of{i}$.
   Sometimes we go even one step further and  say that $\bar{t}$ satisfies some formula $\psi$ whose variables are among $\{x_1,\dots, x_n\}$, hereby implicitly referring to the entries of $\bar{t}$ and the underlying structure.
   We only do so if it improves the readability of the text.

For variables $x_1,\dots, x_n$, let $\NAE(x_1,\dots,x_n)$ be a shortcut for the formula
$  \big(\bigvee\nolimits_{i\neq j} x_i\neq x_j\big)$ (``not all equal''). 
Note that the relation defined by $\NAE(x_1,\dots,x_n)$ is preserved by $\pp$, because it can be equivalently defined by $\big(\bigvee\nolimits_{i>1} x_1\neq x_i\big)$, which is of the form \eqref{eq:CNF}.
Also note that the formula $(x_1\geq x_2\vee \dots \vee x_1\geq x_n)$ can be written as $\big(x_1\geq \min(x_i\mid i>1)\big)$, and its negation $(x_1< x_2\wedge \dots \wedge x_1< x_n)$ can be written as $\big(x_1< \min(x_i\mid i>1)\big)$.
A quantifier-free formula $\phi$ in CNF is \emph{reduced} if it is not possible to remove any disjunct from a clause of $\phi$ so that the resulting formula still defines $R$.
E.g., $\big(\bigvee\nolimits_{i>1} x_1\neq x_i\big)$ is reduced, while $  \big(\bigvee\nolimits_{i\neq j} x_i\neq x_j\big)$ is not reduced for $n\geq 3$.

An \emph{automorphism} is a bijective endomorphism whose inverse is also an endomorphism.
In analogy to Proposition~\ref{prop:reductions}, we have the following.
\begin{proposition}[\cite{Hodges}]  \label{prop:InvAut}  Let $\struct{A}$ be a structure and $R\subseteq A^k$ for some $k\in \mathbb{N}$. If $R$ has a first-order definition in $\struct{A}$, then it is preserved by all automorphisms of $\struct{A}$.  
\end{proposition}  
A \emph{substructure} of $\struct{A}$ is a structure $\struct{B}$ with the same signature and domain $B\subseteq A$ such that $R^{\struct{B}}=R^{\struct{A}}\cap B^k$ for every relation symbol $R$ of arity $k$.
In what follows we will frequently use the fact that the structure $(\mathbb{Q};<)$ is \emph{homogeneous}~\cite{Hodges}, i.e., that every isomorphism between two substructures can be extended to an automorphism.

\subsection{Proof of Lemma~\ref{lemma:ElimMin}} 
 \ElimMin*
 \begin{proof} 
   Let $f(x,y)\coloneqq\lex(\max(x,y), \lex(x,y))$.   
   Clearly,  $f$  is injective and preserves $\leq$.
   Thus, by Proposition~\ref{prop:ordhorn}, $f$ preserves $R$. 
   Let $j_1, j_2 \in [\ell]$ be distinct.
   Suppose that there exist tuples $\bar{t}_{1},\bar{t}_{2}\in S$ such that, for every $i\in \{1,2\}$, $\bar{t}_i$ satisfies 
   $$(x\geq z_{j_i}) \wedge \Big(\bigwedge\nolimits_{j\in [\ell]\setminus \{j_i\}}  x < z_{j}\Big) \wedge \neg \psi_2.$$
   Let $\alpha$ be any automorphism of $(\mathbb{Q};<)$  with $\alpha(\bar{t}_{2}\of{x})=\bar{t}_{1}\of{x}$.
   Then $\bar{t}\coloneqq f(\bar{t}_1, \alpha\bar{t}_2) \in R$  satisfies neither 
   $(x \geq z_1 \vee \cdots \vee x \geq z_{\ell})$ nor $\psi_2$, a contradiction to the existence of both $\bar{t}_1,\bar{t}_{2}$ at the same time.
   Thus, $(x \geq z_1 \vee \cdots \vee x \geq z_{\ell})$ can be replaced by $(x \geq z_1 \vee \cdots x \geq z_{j-1}\vee x \geq z_{j+1} \vee \cdots \vee x \geq z_{\ell})$ for some $j\in \{j_1,j_2\}$ without changing $R$.
   Iterating this $\ell-1$ times, until we eliminate all but one index from $[\ell]$, yields $(x\geq z_i)$ for some $i\in [\ell]$. 
\end{proof} 


\subsection{Proof of Lemma~\ref{lemma:pp-def}}
  \PPDef*
  \begin{proof} Let $\theta\coloneqq \exists h \ \phi_{k}(x,y_1,\dots,y_{k},h)  \wedge \Michal(h,h,x) \wedge \Michal(h,y_{k+1},z_1)$. 
  
  ``$\Rightarrow$'' Suppose that $\bar{t}\in \mathbb{Q}^{k+3}$ satisfies $\mu_{k+1}$.
  If $\bar{t}\of{x}=\bar{t}\of{y_1}=\cdots=\bar{t}\of{y_{k}}$, then we choose  $h\coloneqq \bar{t}\of{x}$ to satisfy $\theta$.
  Otherwise, we choose any $h>\max\{\bar{t}\of{y_1},\dots, \bar{t}\of{y_{k+1}}\}$ to satisfy $\theta$.

  ``$\Leftarrow$'' Suppose that $\bar{t}\in \mathbb{Q}^{k+3}$ satisfies $\theta$.
   If it is not the case that $\bar{t}\of{x}=\bar{t}\of{y_1}=\cdots=\bar{t}\of{y_{k}}$, then clearly $\bar{t}$ satisfies $\mu_{k+1}$.
   Otherwise, by the first two conjuncts in $\theta$, it must be the case that $h=\bar{t}\of{x}$.
   If $\bar{t}\of{x}\neq \bar{t}\of{y_{k+1}}$, then clearly $\bar{t}$ satisfies $\mu_{k+1}$.
   Otherwise $\bar{t}\of{x}=h\geq \bar{t}\of{z_1}$ because of the last conjunct.
  \end{proof}

\subsection{Proof of Lemma~\ref{lemma:algo_p_halt}\label{appendix:lemma:algo_p_halt}} 

\AlgoHalt*

In the proof of Lemma~\ref{lemma:algo_p_halt}, we frequently use the following simple observation, without explicitly mentioning it. 
\begin{claim} \label{obs:u-vs-A}
For every pair $(x,z)\in V^2$, and every $A \subseteq \varsu$, there exists $u \in V$ such that $A \subseteq {\uparrow_u}$ and ${\uparrow_A}\setminus (\{x,z\}\cup \cut{x}{z}) = {\uparrow_u} \setminus (\{x,z\}\cup \cut{x}{z}).$
\end{claim} 
\begin{claimproof}
It is easy to see that if $A\neq \emptyset$, then we may choose $u$ to be the variable in $A$ that satisfies $u \preceq y$ for all $y \in A$.

If $A=\emptyset$, then we choose $u$ as the last variable in the quantifier-prefix of $\Phi$. 
Indeed, if $u$ is existential, then we are done. Otherwise, $u$ is universal. If $u \in \{ x,z\}$, then this variable is removed from $\uparrow_A$ and we are done.  If $u \notin \{x,z\}$, then $u \in \cut{x}{z}$. This completes the proof of the 
observation.   
\end{claimproof}

\begin{proof}[Proof of Lemma~\ref{lemma:algo_p_halt}] 
We assume that $\phi$ is expanded by all derived constrains from the run of Algorithm~\ref{algo:dima}, and show that~$\big((\bigwedge\nolimits_{v\in \uparrow_A\setminus (\{x,z\}\cup \cut{x}{z})} x=v) \Rightarrow x\geq z \big)$ is among these constraints.
We prove the lemma by induction on the length of the derivation of $\mathcal{P}(x,z;A)$. Observe that it is enough to show that,  if $\mathcal{P}(x,z;A)$ is derived, where $z \notin A$, then $\phi \wedge (\bigwedge_{v\in A} x=v) \wedge (x < z)$ is not satisfiable. Then indeed, by Observation~\ref{obs:u-vs-A}, we may choose $u \in V$ such that $A \subseteq {\uparrow_u}$ and ${\uparrow_A}\setminus (\{x,z\}\cup \cut{x}{z}) = {\uparrow_u} \setminus (\{x,z\}\cup \cut{x}{z})$. Since $z \notin A$ and $(x=x)$ is always satisfied, if $\phi \wedge (\bigwedge_{v\in A} x=v) \wedge (x < z)$ is not satisfiable, then neither is $\phi \wedge (\bigwedge_{v\in \uparrow_u \setminus \{x, z\}} x=v) \wedge (x < z)$, and therefore the algorithm expands $\phi$ by the desired conjunct. 

 \Init:  In the base case, we consider an expression of the form  $\mathcal{P}(x,x;\emptyset)$ for $x\in V$. Since $(x < x)$ is always an unsatisfiable constraint, the claim holds true. 

\Simplify: Suppose that $\mathcal{P}(x,z;A\setminus \cut{x}{z})$ was derived from $\mathcal{P}(x,z;A)$, where $z\notin A\setminus \cut{x}{z}$.  
Then also $z\notin A$, because $z\notin \cut{x}{z}$, and so we can apply the induction hypothesis to $\mathcal{P}(x,z;A)$.
 By the induction hypothesis, $\phi \wedge (\bigwedge_{v\in A} x=v) \wedge (x < z)$ is not satisfiable and hence, as explained above, the algorithm has expanded $\phi$ by $\big((\bigwedge_{v\in \uparrow_A\setminus (\{x,z\}\cup \cut{x}{z})} x=v) \Rightarrow x\geq z\big)$.
 Therefore, since $A \setminus \cut{x}{z} \subseteq {\uparrow_A} \setminus (\{z\} \cup \cut{x}{z})$ and $(x=x)$ holds trivially, we obtain that
 $\phi \wedge (\bigwedge_{v\in A \setminus \cut{x}{z}} x=v) \wedge (x < z)$ is not satisfiable, as we needed. 

\Trans: Suppose that we derived $\mathcal{P}(x,z;A)$ from $\mathcal{P}(x,y;A)\wedge \mathcal{P}(y,z;\emptyset)$, where $z\notin A$.
 By the induction hypothesis, $\phi \wedge (y < z)$ has no solution and the algorithm has expanded $\phi$ by $(y\geq z)$.
 If $y\in A$, then clearly $\phi \wedge (\bigwedge_{v\in A} x=v) \wedge (x<z)$ is not satisfiable because $(y=x)\wedge (x < z)$ entails $(y < z)$ .

 Otherwise $y \notin A$. Then the induction hypothesis is applicable also on $\mathcal{P}(x,y;A)$ and hence $\phi \wedge (\bigwedge_{v\in A} x=v) \wedge (x<y)$ has no solution.
 Since $\phi$ contains $(y \geq z)$,
 $\phi \wedge (\bigwedge_{v\in A} x=v) \wedge (x<z)$ has no solution, which completes this case.

\AltTrans:  Suppose that we derived $\mathcal{P}(x_i,z;A\cup B\cup (\{x_1,x_2\}\setminus \{x_i\}) )$ from 
\begin{enumerate}
    \item $\mathcal{P}(x_1,y;A) \wedge \mathcal{P}(y,x_2;\emptyset)\wedge\mathcal{P}(y,z;B)$,
    \item $(\{x_1,x_2\}\setminus \{x_i\})\subseteq \varsu$ ($i\in \{1,2\}$),
\end{enumerate} 
 where $z\notin A\cup B\cup (\{x_1,x_2\}\setminus \{x_i\})$.%
By the induction hypothesis, $\phi \wedge (y < x_2)$ has no solution and the algorithm has expanded $\phi$ by $(y\geq x_2)$.
Similarly, applying the induction hypothesis to $\mathcal{P}(y,z;B)$, $\phi \wedge (\bigwedge_{v\in B} y=v) \wedge (y<z)$ has no solution.
If $y\in A$, then clearly 
$\phi \wedge \big(\bigwedge\nolimits_{v\in (A\cup B \cup (\{x_1,x_2\}\setminus \{x_i\}) ) } x_i=v \big) \wedge (x_i<z)$
is not satisfiable, otherwise it would violate the induction assumption for $\mathcal{P}(y,z;B)$.
Assume therefore that $y\notin A$ and hence the induction hypothesis is applicable on $\mathcal{P}(x_1,y; A)$.
It follows that $\phi \wedge (\bigwedge_{v\in A} x_1=v) \wedge (x_1<y)$ has no solution.

 Suppose for contradiction, that
 $\phi \wedge \big( \bigwedge\nolimits_{v\in A\cup B \cup (\{x_1,x_2\}\setminus \{x_i\})} x_i=v \big) \wedge (x_i<z)$
 is satisfiable. Then we must have $(y \geq x_2=x_1)$ and since $(x_1 <y)$ cannot be satisfied along with this formula, $(y=x_2=x_1)$. Moreover, $(y < z)$ cannot be satisfied, therefore $(x_1=x_2=y \geq z)$. This yields a contradiction with $(x_i <z)$, therefore the formula is not satisfiable, which completes the case.

\Progress:
Suppose that  $ \mathcal{P}(x_i,z;A\cup B\cup \{x_1, x_2, x_3, x_4\} \setminus \{x_i \}) $ was derived from 
\begin{enumerate}
    \item $\mathcal{P}(x_1,u;A)  \wedge   \mathcal{P}(u,x_2;\emptyset)$,
    \item $ \mathcal{P}(x_3,v;B)\wedge   \mathcal{P}(v,x_4;\emptyset)$,
    \item $(\{x_1,x_2,x_3,x_4\} \setminus \{x_i\})\subseteq \varsu$ ($i\in \{1,2,3,4\}$),   
    \item $(u=v \Rightarrow u \geq z)$ or $(v=u \Rightarrow v \geq z)$ in $\phi$,
\end{enumerate} 
 where $z\notin A\cup B\cup \{x_1, x_2, x_3, x_4\} \setminus \{x_i \}$.
By the induction hypothesis, $\phi \wedge (u < x_2)$ and $\phi \wedge (v < x_4)$ are not satisfiable and the algorithm has expanded $\phi$ by $(u\geq x_2)\wedge (v\geq x_4)$.
Next we consider the four cases $u\in A$ and $v\in B$, $u\notin A$ and $v\in B$, $u\in A$ and $v\notin B$, $u\notin A$ and $v\notin B$.

If $u \in A$ and $v \in B$, then 
$\phi \wedge \big(\bigwedge\nolimits_{w\in A\cup B \cup (\{x_1,x_2,x_3,x_4\}\setminus \{x_i\})} x_i=w\big)$
entails $(x_1=x_2=x_3=x_4=u=v \geq z)$, which implies that
$\phi \wedge \big(\bigwedge\nolimits_{w\in A\cup B \cup (\{x_1,x_2,x_3,x_4\}\setminus \{x_i\})} x_i=w\big) \wedge (x_i < z)$
has no solution.

If $u \notin A$ and $v \in B$, then we apply the induction hypothesis to $\mathcal{P}(x_1, u; A)$ and get that
$ \phi \wedge \big(\bigwedge\nolimits_{w\in A} x_1=w\big) \wedge (x_1 < u)$ has no solution. Therefore,
$\phi \wedge \big(\bigwedge\nolimits_{w\in A\cup B \cup (\{x_1,x_2,x_3,x_4\}\setminus \{x_i\})} x_i=w\big)$
entails $(x_1=x_2=x_3=x_4=v \geq u)$ and $(u \geq x_2)$.
Since $\phi$ also contains $(u=v \Rightarrow u \geq z)$ or $(v=u \Rightarrow v \geq z)$,  we get that  
$\phi \wedge \big(\bigwedge\nolimits_{w\in A\cup B \cup (\{x_1,x_2,x_3,x_4\}\setminus \{x_i\})} x_i=w\big) \wedge (x_i < z)$
has no solution. The case $u \in A$ and $v \notin B$ is symmetric. 
Finally, suppose that $u \notin A$ and $v \notin B$. 
Then we may apply the induction hypothesis on both $\mathcal{P}(x_1,u; A)$ and $\mathcal{P}(x_3,v;B)$ and assume that neither $ \phi \wedge \big(\bigwedge\nolimits_{w\in A} x_1=w\big) \wedge (x_1 < u)$ nor $ \phi \wedge \big(\bigwedge_{w\in B} x_3=w\big) \wedge (x_3 < v)$ has a solution. 
Therefore, 
$\phi \wedge \big(\bigwedge\nolimits_{w\in A\cup B \cup (\{x_1,x_2,x_3,x_4\}\setminus \{x_i\})} x_i=w\big)$
entails  $(x_1=x_2=x_3=x_4)$, $(x_1\geq u)$, $(x_3 \geq v)$, $(u \geq x_2)$  and $(v \geq x_4)$. In particular, it entails $(u=v)$. Since $\phi$ also contains $(u=v \Rightarrow u \geq z)$ or $(v=u \Rightarrow v \geq z)$, we obtain again that
$\phi \wedge \big(\bigwedge\nolimits_{w\in A\cup B \cup (\{x_1,x_2,x_3,x_4\}\setminus \{x_i\})} x_i=w\big) \wedge (x_i < z)$
has no solution.
This finishes the proof of the lemma.
\end{proof}

 \subsection{Proof of Claim~\ref{claim:charact_game}\label{appendix:claim:charact_game}}

\CharactGame*

\begin{claimproof}  ``$\Leftarrow$''
We show that $\ass{x}=\ass{x_2}$ and $\ass{z}=\ass{z_2}$.
If $x\equiv x_2$, then clearly $\ass{x}=\ass{x_2}$. So, w.l.o.g., $x_2\prec x$.  
If $x\in \varsu$, then \Simplify yields $\mathcal{P}(x_2,x;\emptyset)$ and hence \Refute produces $\bot$, a contradiction. 
So we must have $x\in \varse$. Then either $x_1\equiv  x$ or $x_1 \prec x$, and it follows from the strategy of the EP that $\ass{x}=\ass{x_2}$.
Analogously we obtain that
$\ass{z}=\ass{z_2}$.
The rest follows by the transitivity of the equality.

``$\Rightarrow$'' Whenever the  right-hand side of the equivalence in Claim~\ref{claim:charact_game} is satisfied, we  call   $(x,x_1,x_2;A_{x_2,x})$ and $(z,z_1,z_2;A_{z_2,z})$ \emph{witnessing quadruples} for $\ass{x}=\ass{z}$.
If $x\equiv z$, then the statement trivially follows using \Init ,  the witnessing quadruples are $(x,x,x;\emptyset)$ and $(z,z,z;\emptyset)$.
So, w.l.o.g., $z\prec x$.
If $x\in \varsu$, then the claim follows using \Init , the witnessing quadruples are again $(x,x,x;\emptyset)$ and $(z,z,z;\emptyset)$.  
So suppose that $x\in \varse$ and  that the claim holds for all pairs of variables preceding $x$.
Since $\ass{x}=\ass{z}$, by the strategy of the EP, there exist $x_1,x_2\prec x$ and $x_{2}\prec A \prec x$ such that $\mathcal{P}(x,x_1;\emptyset)\wedge \mathcal{P}(x_2,x;A)$ and   $\ass{z}=\ass{x_1}=\ass{x_2}=\ass{A}.$ 

Since $\ass{x_2}=\ass{z}$ and $x_2,z\prec x$, we can apply the induction hypothesis for the pair $x_2,z$ to obtain the witnessing quadruples $(x_2,x_{2_1},x_{2_2};A_{x_{2_2},x_2})$ and $(z,z_{1},z_{2};A_{z_{2},z})$.
By assumption, there exists $y\in \{z_2,x_{2_2}\}$ such that $y\preceq_{\forall} \{z_1,z_2,x_{2_1},x_{2_2}\}$.
Note that $\ass{y}=\ass{x_1}$.
Thus, we can apply the induction hypothesis for the pair $x_1,y$ to obtain the witnessing quadruples $(x_1,x_{1_1},x_{1_2};A_{x_{1_2},x_1})$ and $(y,y_{1},y_{2};A_{y_{2},y})$.  
By assumption, there exists $y'\in \{y_2,x_{1_2}\}$  such that $y'\preceq_{\forall} \{y_1,y_2,x_{1_1},x_{1_2}\}$.  
The two cases w.r.t.\ the variable $y$ that can occur are illustrated in Figure~\ref{figure:cases}, see Case~1 and Case~2 below.

  Our goal is to find witnesses $x'_1,x'_2,z'_1,z'_2$ for the main statement of the claim, i.e., the witnessing quadruples will be of the form $(x,x'_1,x'_2;A_{x'_2,x})$ and $(z,z'_1,z'_2;A_{z'_2,z})$. 
For the sake of brevity, we will not explicitly write down the precise definitions of 
$A_{x'_2,x}$ and $A_{z'_2,z}$ as they will be clear from the context.
We set $x'_1\coloneqq x_{1_1}$, and: 
 \[ z'_1\coloneqq \begin{cases} y_1 & \text{if } y\equiv z_1, \\ z_1 & \text{otherwise,}
 \end{cases} 
  \qquad  z'_2\coloneqq \begin{cases} y_2 & \text{if } y\equiv z_2, \\ z_2 & \text{otherwise,}
 \end{cases}
 \qquad 
   x'_2\coloneqq \begin{cases} y_2 & \text{if } y\equiv x_{2_2}, \\ x_{2_2} & \text{otherwise.}
 \end{cases} \qquad 
  \]   

By the induction hypothesis and the transitivity of $\prec$, $x_1',x_2',z_1',z_2'$ clearly satisfy item~\ref{item:ugly_claim_1}.
It will also be clear that our implicit choice of $A_{x'_2,x}$ and $A_{z'_2,z}$ leads to satisfaction of item~\ref{item:ugly_claim_5}.  
The remaining three items are proved in the case distinction below.
In both Cases~1 and~2, we initially start proving items~\ref{item:ugly_claim_2} and~\ref{item:ugly_claim_4}, and then proceed with item~\ref{item:ugly_claim_3} in the finer subdivision into Cases~1.1, 1.2, 2.1, and~2.2. For the sake of conciseness, when applying rules of $\mathcal{P}$ to derive new expressions, we often do not state all necessary expressions for the inference, as long as they are clear from the rule and the resulting expression.

 To justify the applications of the rule \AltTrans  that follow, we observe that $y_2 \preceq_{\forall} y_1$. Indeed, suppose that $y_1 \not\equiv y_2$. By \Trans , we have $\mathcal{P}(y_2, y_1, A_{y_2,y})$. If $y_1 \in \varse$ or $y_1 \prec y_2$, then  $y' \preceq_{\forall} \{y_1, y_2\}$ implies $y_1 \prec y_2$ and $y_2 \in \varsu$. By \Simplify, we obtain $\mathcal{P}(y_2, y_1, \emptyset)$, and then using \Refute we get $\bot$, a contradiction.
 By an analogous argument, we have $x_{2_2} \preceq_{\forall} x_{2_1}$.  
  Now it immediately follows that,  by \AltTrans , we have 
\begin{align} \mathcal{P}(x_{2_2}, x; A \cup  A_{x_{2_2},x_2} \cup  (\{x_{2_1}\} \setminus \{x_{2_2}\}).
      \label{eq:I_am_stupid3} 
\end{align}
A subsequent double application of \Trans yields 
\begin{align} \mathcal{P}(x_{2_2},x_{1_1};A\cup A_{x_{2_2},x_2}\cup (\{x_{2_1}\}\setminus \{x_{2_2}\})).\label{eq:I_am_stupid4} 
\end{align}   

\case{1}{$y\equiv z_2$}. Then $z_2 \preceq_{\forall} \{z_1,x_{2_1},x_{2_2}\}$.
We now establish item 4 with a suitable choice of sets $A_{x'_2,x}$ and $A_{z'_2,z}$. It is easy to verify that item 2 is satisfied as well.
\begin{itemize}
\item  If $y\equiv z_1$, then, by \Trans , we have $\mathcal{P}(z,z'_1;\emptyset)$ because $z'_1\equiv y_1$. 
Otherwise $z_1\in \varsu$ and $y\prec z_1$; we have $\mathcal{P}(z,z'_1;\emptyset)$, because $z'_1\equiv z_1$.

\item Recall that $y_2\preceq_{\forall} y_1$.
  Thus, by \AltTrans , we have \begin{align}  \mathcal{P}(z'_2,z;A_{z_2,z}\cup A_{y_2,y}\cup (\{y_1\}\setminus\{y_2\})),
      \label{eq:I_am_stupid5} 
\end{align}   because $z'_2\equiv y_2$.

\item By \Trans, we have $\mathcal{P}(x,x'_1;\emptyset)$, because 
$x'_1=x_{1_1}$.

\item  By assumption, we have $y\preceq_{\forall} x_{2_2}$.
Recall that we have \eqref{eq:I_am_stupid3}.
If $y\equiv x_{2_2}$, then an application of \AltTrans yields
\begin{align}
    \mathcal{P}(x'_{2},x;A \cup A_{x_{2_2},x_2} \cup   (\{x_{2_1}\} \setminus \{x_{2_2}\})\cup A_{y_2,y}\cup (\{y_1\}\setminus\{y_2\})),   \label{eq:I_am_stupid} 
\end{align} 
because $x'_{2}\equiv y_2$.
Otherwise $y \prec x_{2_2}$ and $x_{2_2}\in \varsu$.
%
%
Then $\mathcal{P}(x'_{2},x;A\cup A_{x_{2_2},x_2} \cup (\{x_{2_1}\}\setminus\{x_{2_2}\}))$ follows directly from \eqref{eq:I_am_stupid3} because $x'_{2}=x_{2_2}$.
\end{itemize}

In the case distinction below, we verify item \ref{item:ugly_claim_3} for $x_1',x_2',z_1',z_2'$.

\case{1.1}{$y'\equiv y_2$}. Since $y_2\preceq y\equiv z_2$, in all cases above we get $z'_2 \equiv  y_2 \preceq_{\forall} \{z'_1,x'_1,x'_2\}$. 

\case{1.2}{$y'\equiv x_{1_2}$}. 
We show that, as above, $z'_2 \equiv  y_2 \preceq_{\forall} \{z'_1,x'_1,x'_2\}$, starting with $x'_2$. 
If $y\equiv z_2\prec x_{2_2}$ and $x_{2_2}\in \varsu$, then we have $z'_2 \equiv  y_{2}\preceq_{\forall} x_{2_2}\equiv x'_2$.
Otherwise $x_{2_2} \equiv z_2 \equiv  y$ and by our choice of $x'_2$ and $z'_2$, we have $z'_2\equiv y_{2}$ and $x'_2\equiv y_{2}$. In particular, $z'_2  \preceq_{\forall} x'_2$. 
%
%
Next comes $z'_1$. 
If $y\prec z_1$ and $z_1\in \varsu$, 
then $y_2\preceq_{\forall} z_1$ because $y_2\preceq y$.
Consequently, $z'_2 \equiv y_2\preceq_{\forall} z_1\equiv z'_1$.
Otherwise $y\equiv z_1$ and $z'_1 \equiv y_1$. 
Recall that we always have $y_2 \preceq_{\forall} y_1$ and hence $z'_2\equiv y_2 \preceq_{\forall} y_1 \equiv  z'_1$. 

Finally $x'_1$. Recall that we have $z_{2} \preceq_{\forall}  x_{2_2} \preceq_{\forall} x_{2_1}$.
We consider the following two cases. First, suppose that $z_2\prec x_{2_2}$ and $x_{2_2}\in \varsu$.
Since $x_{2_2}\in \varsu$, it cannot be the case that $x_{1_1}\prec x_{2_2}$, otherwise \Simplify applied on \eqref{eq:I_am_stupid4} yields $\mathcal{P}(x_{2_2},x_{1_1};\emptyset)$. Using \Refute, we obtain $\bot$, a contradiction. 
Hence $x_{2_2}\preceq x_{1_1}$.
Now it follows that $y_2\preceq z_2\prec x_{2_2}\preceq x_{1_1}$.
Since $y'\preceq_{\forall} \{y_2,x_{1_1}\}$,
we even have $z'_2\equiv y_{2}\preceq_{\forall} x_{1_1}\equiv x'_1$.
Second, suppose that $x_{2_2} \equiv z_2 \equiv y$. 
Recall that $y_2\preceq_{\forall} y_1$ and $\mathcal{P}(y_2,y;A_{y_2,y})$.
Combining this with \eqref{eq:I_am_stupid4} and applying
\AltTrans, we get 
\begin{align}\mathcal{P}(y_{2},x_{1_1};A\cup A_{x_{2_2},x_2}\cup (\{x_{2_1}\}\setminus \{x_{2_2}\})\cup A_{y_2,y}\cup (\{y_{1}\}\setminus \{y_{2}\})). 
      \label{eq:I_am_stupid2} 
\end{align}  
We cannot have $x_{1_1}\prec y_2$, otherwise $y_2 \in \varsu$ in which case \Simplify yields $\mathcal{P}(y_{2},x_{1_1};\emptyset)$ and then \Refute yields $\bot$, a contradiction. 
Hence, $y_2\preceq  x_{1_1}$.
Since $y'\preceq_{\forall} \{y_2,x_{1_1}\}$,
we get $z'_2\equiv y_2\preceq_{\forall} x_{1_1}\equiv x'_1$. 

\case{2}{$y\equiv x_{2_2}$}. Then $x_{2_2} \preceq_{\forall} \{z_1,z_2,x_{2_1}\}$.
We now show that item 4 holds true with a suitable choice of sets $A_{x'_2,x}$ and $A_{z'_2,z}$; it will be clear that item 2 is satisfied as well.
\begin{itemize}
\item If $y\equiv z_1$, then, by \Trans , we get $\mathcal{P}(z,z'_1;\emptyset)$, because $z'_1\equiv y_1$.
Otherwise $z_1\in \varsu$ and $y\prec z_1$; we have $\mathcal{P}(z,z'_1;\emptyset)$ because $z'_1\equiv z_1$. 

\item First, suppose that  $y\equiv z_2$. 
Recall that we have $y_2\preceq_{\forall} y_1$.
By \AltTrans , we have \eqref{eq:I_am_stupid5} because $z'_2 \equiv y_2$.
Second, suppose that $y\prec z_2$ and $z_2\in \varsu$. Then we have $\mathcal{P}(z'_2,z;A_{z_{2},z})$ because $z'_2 \equiv z_2$.

\item By \Trans , we have $\mathcal{P}(x,x'_1;\emptyset)$ because $x'_1\equiv x_{1_1}$.

\item Recall that we have \eqref{eq:I_am_stupid3} and $y_2\preceq_{\forall} y_1$.
By \AltTrans , we have \eqref{eq:I_am_stupid} because $x'_{2}\equiv y_2$.
\end{itemize}

Finally, we verify item 3 of the claim.

\case{2.1}{$y'\equiv y_2$}. Since $y_2\preceq y\equiv x_{2_2}$, in all cases above we get $x'_2 \equiv  y_2 \preceq_{\forall} \{z'_1,z'_2,x'_1\}$.

\case{2.2}{$y'\equiv x_{1_2}$}. By a double application of \Trans on \eqref{eq:I_am_stupid}, we get \eqref{eq:I_am_stupid2}.  
It cannot be that $x_{1_1}\prec y_2$, as this implies $y_2\in \varsu$, in which case \Simplify yields $\mathcal{P}(y_2,x_{1_1};\emptyset)$ and then \Refute yields $\bot$, a contradiction as in Case~1.2. 
Hence $x'_2\equiv y_2\preceq_{\forall} x_{1_1}\equiv x'_1$.
Recall that $ y_2\preceq_{\forall} y_1$.
Either $x_{2_2}\prec z_1$ and $z_1\in \varsu$, in which case $x'_2\equiv y_2\preceq_{\forall} z_1 \equiv z'_1 $, or $x_{2_2}\equiv z_1$ in which case $x'_2\equiv y_2\preceq_{\forall} y_1 \equiv  z'_1$. 
Also either $x_{2_2}\prec z_2$ and $z_2\in \varsu$, in which case $x'_2\equiv y_2\preceq_{\forall} z_2 \equiv z'_2 $, or $x_{2_2}\equiv z_2$ in which case $x'_2\equiv y_2\preceq_{\forall} y_2 \equiv  z'_2$.
Hence, $x'_2  \preceq_{\forall} \{z'_1,z'_2,x'_1\}$.
\end{claimproof} 

 \subsection{Proof of Lemma~\ref{lemma:arity_four}\label{appendix:lemma:arity_four}}
\ArityFour*

   The  proof of Lemma~\ref{lemma:arity_four} is based on the notion of connectedness within OH clauses.
   Let $\psi$ be an OH clause, i.e., a formula of the form \eqref{eq:OrdHorn}.
    We say that two variables are \emph{connected} in $\psi$ if there exists a walk between them over the disjuncts in $\psi$ via the variables that are identified in the disjuncts.
    E.g., if $\psi$ is obtained from  \eqref{eq:OrdHorn} for $k=2$ by identifying $y_1$ and $x_2$, then $x_1$ is connected with $y_2$ but not with $y_3$. A clause is \emph{connected} if all its variables are connected. 

\begin{proof}[Proof of Lemma~\ref{lemma:arity_four}] 
    Let $R$ be a relation of $\struct{B}$ that is not  preserved by $\pp$.
    Let $\phi$ be an OH definition of $R$, i.e., a conjunction of clauses of the form  \eqref{eq:OrdHorn}  
    where the last disjunct is optional and some variables might be identified.
    We prove Lemma~\ref{lemma:arity_four} by induction on the following two parameters: in each step, we either
    \begin{itemize}
       \item obtain a new pp-definable relation of strictly smaller arity that is not preserved by $\pp$, or
        \item  leave the relation from the last step intact but provide a new OH definition for it so that the sum of the numbers of connected components of all its clauses is strictly smaller.
    \end{itemize} 
    It is important to note that, by Proposition~\ref{prop:ordhorn} and Proposition~\ref{prop:reductions}, every temporal relation pp-definable using an OH relation is again OH.
    Clearly, after finitely many steps, the above process yields a pp-definable relation of arity $\leq 4$ that is not preserved by $\pp$.

    \medskip We now proceed with the induction step. Since $R$ is not preserved by $\pp$, some clause 
    $\psi$ of the form \eqref{eq:OrdHorn}
    is \emph{violated} by $\pp$, i.e., there exist tuples $\bar{t}_1,\bar{t}_2\in R$ such that $\pp(\bar{t}_1,\bar{t}_2)$ does not satisfy $\psi$.
    In what follows we identify $\psi$ with the formula in \eqref{eq:OrdHorn} where some variables might be identified.
    We may assume without loss of generality that the variables of $\psi$ cover all entries of $R$, otherwise we consider a projection of $R$ to the variables of $\psi$, which is a relation with an Ord-Horn definition containing $\psi$. 
    We may further assume that $\phi$ is reduced.
    Note that, since $\psi$ is reduced, it cannot be the case that, 
    \begin{itemize}
        \item for some $i\in [k+1]$, the variables $x_i$ and $y_i$ are identified in $\psi$,
        \item for some $i\in [k]$, the variables $\{x_i,y_i\}$ are identified with the variables $\{x_{k+1},y_{k+1}\}$.
    \end{itemize}
    If the arity of $R$ is $\leq 4$, then we are in the base case and there is nothing to show. So suppose that the arity of $R$ is $> 4$.

\case{1}{$\psi$ is connected}.
   If the last disjunct is not present in $\psi$, then $\psi$ is logically equivalent to $\NAE(x_i,y_i \mid i \in [k])$,
   which is preserved by $\pp$,
   a contradiction.
   So the last conjunct must be present in $\psi$.
   If $x_{k+1}$ is identified with some of the other variables, then $\psi$ is equivalent to a clause of the form \eqref{eq:CNF} and again defines a relation preserved by $\pp$, a contradiction.
   So $x_{k+1}$ must be distinct from the remaining variables.
   By connectedness of $\psi$, we may without loss of generality assume that $\psi$ is of the form   
   $  
    (z_{\ell+1} \geq z_1) \vee \NAE(z_1,\dots,z_\ell),
   $
   where $\{z_1,\dots, z_{\ell+1}\}\subseteq \{x_1,y_1,\dots, x_{k+1},y_{k+1}\}$ and $\ell \geq 4$ because the arity of $R$ is greater than $4$. 
   Since $\phi$ is in reduced CNF, all variables $z_1, \dots, z_{\ell+1}$ are distinct
   and there exists a tuple $\bar{s}\in R$ satisfying  $(z_{\ell +1} \geq z_{1} =\cdots = z_{\ell})$. 

 We claim that there exists a tuple $\bar{t}\in R$ satisfying  $\big(z_{\ell +1}< \min(z_i\mid i\in [\ell])\big)$.
   To see this, suppose this is not the case.
   Then $\phi$ entails $\big(z_{\ell +1}\geq \min(z_i\mid i\in [\ell])\big)$. 
   Note that whenever $\big(z_{\ell +1}\geq \min(z_i\mid i\in [\ell])\big)$ holds,  $\psi$ holds as well.  
   Since $\pp$ preserves $\big(z_{\ell +1}\geq \min(z_i\mid i\in [\ell])\big)$, the fact that $R$ is not preserved by $\pp$ cannot be witnessed by violating $\psi$, a contradiction.
  Thus there exists a tuple $\bar{t}\in R$ satisfying $\big(z_{\ell +1}< \min(\bar{t}\of{z_i}\mid i\in [\ell])\big)$.   

  For every such $\bar{t}$ we can select an automorphism $\alpha$ of $(\mathbb{Q};<)$ with $\alpha \bar{t}\of{z_{\ell +1}} = 0$. Then it is easy to see that $\pp(\alpha\bar{t},\bar{s})$ violates $\psi$.
   Recall that the variables $z_1$ and $z_2$ are distinct in $\psi$. 
   We now make a case distinction based on the (non)-existence of a tuple $\bar{t}'\in R$ satisfying
    \begin{align}
     \big(z_{\ell +1} < \min(z_i\mid i\in [\ell])\big) \wedge (z_{1} = z_{2}).
        \label{eq:special_condition}
    \end{align}

   \case{1.1}{there exists $\bar{t}'\in R$ satisfying~\eqref{eq:special_condition}}.
   Then existentially quantifying $z_1$ and adding $(z_1=z_2)$ as a conjunct yields a pp-definable relation $R'$ of a strictly smaller arity that is not preserved by $\pp$, since $\pp(\alpha \bar{t}', \bar{s})$ does not satisfy $\psi$. 
   This completes the induction step. 

   \case{1.2}{no $\bar{t}'\in R$ satisfies~\eqref{eq:special_condition}}. 
   Then $\phi$ entails $\big(z_1\neq z_2 \vee z_{\ell +1}\geq \min(z_1,\dots,z_{\ell})\big)$.
  By Lemma~\ref{lemma:ElimMin}, this clause can be replaced by $(z_1\neq z_2 \vee z_{\ell+1}\geq z_i)$ for some $i\in [\ell]$. 
   It is easy to see that $\pp(\alpha\bar{t},\bar{s})$ violates
   $(z_1\neq z_2 \vee z_{\ell +1}\geq z_i)$
   for every $\alpha\in \Aut(\mathbb{Q};<)$ satisfying $\alpha \bar{t}\of{z_{\ell +1}} = 0$.
   Therefore, the projection of $R$ to the entries $\{z_1,z_2, z_i,z_{\ell +1}\}$ is a pp-definable relation of arity $\leq 4$ that is not preserved by $\pp$. This completes the induction step.

\case{2}{$\psi$ is not connected}.
Let $P$ be a partition of $[k+1]$ such that $x_i$ and $x_j$ are connected in $\psi$ if and only if $i,j\in C$ for some $C\in P$.
It is easy to see that $\psi$ is equivalent to the formula
\begin{align*}
 (x_{k+1} \geq y_{k+1})  \vee  \Big(\bigvee\nolimits_{C\in P}\NAE(x_{i},y_{i} \mid i \in C\setminus\{k+1\})  \Big),
\end{align*} 
where the first disjunct is optional.  
For the purpose of this proof, we may assume that the first disjunct  is present, it will be clear that the argument also works if it is not. 
Recall that some variables might be identified, in particular, the above equation does not assert anything about the containment of $k+1$ in any set in $P$. 

 Let $C_1,C_2\in P$ be arbitrary and distinct, such that:
    \begin{itemize}
        \item if possible, we select $C_1, C_2\in P$ so that $|C_1|>1$ or $|C_2|>1$,
        \item if all sets in $P$ are singletons, then we select $C_1,C_2\in P$ so that $k+1\in C_1\cup C_2$.
    \end{itemize} 
Let $i_1\in C_1$,$i_2\in C_2$ be arbitrary such that $i_j \neq k+1$ for both $j \in \{1,2\}$, unless $C_j=\{k+1\}$.   
 
We now make a case distinction based on the (non)-existence of an index $j\in \{1,2\}$ and a tuple $\bar{t}\in R$ that 
\begin{enumerate}
    \item \label{item:building_bridges_1}  does not satisfy the $i$-th disjunct in \eqref{eq:OrdHorn} for every $i\in [k+1]\setminus  C_{3-j}$, 
    \item \label{item:building_bridges_2} satisfies $\big(x_{i_j}<\min(x_i,y_i \mid i\in  C_{3-j})\big)$.
\end{enumerate}  

\case{2.1}{for both $j\in \{1,2\}$, no $\bar{t}\in R$ witnesses items~\ref{item:building_bridges_1} and \ref{item:building_bridges_2}}.
Then, for both $j\in \{1,2\}$, $\phi$ entails
\begin{align*}
     \Big(\bigvee\nolimits_{C\in P\setminus \{ C_{3-j}\}}\NAE(x_{i},y_{i} \mid i \in C ) \Big)     \vee   \big( x_{i_j}\geq \min(x_{i},y_{i} \mid i \in C_{3-j})   \big)
\end{align*} 
if $k+1\in C_{3-j}$, and  
\begin{align*}
     &  \Big(\bigvee\nolimits_{C\in P\setminus \{ C_{3-j}\}}\NAE(x_{i},y_{i} \mid i \in C\setminus\{k+1\})\Big)  \vee   \big(x_{i_j}\geq \min(x_{i},y_{i} \mid i \in C_{3-j})\big)     \\ & {}   \vee (x_{k+1} \geq y_{k+1}) 
\end{align*} 
otherwise,  where we assume the same identification of variables as in $\psi$.
But then  we can replace $\psi$ with the clause $\psi'$ defined by 
\begin{align*}
   &   \Big(\bigvee\nolimits_{C\in P\setminus \{C_1,C_2\}}\NAE(x_{i},y_{i} \mid i \in C\setminus\{k+1\})\Big) \vee    \NAE(\{x_{i},y_{i} \mid i \in C_1\cup C_2\}\setminus\{y_{k+1}\})    
\\ &  {}  \vee     (x_{k+1} \geq y_{k+1})  
\end{align*}
without changing $R$ (while keeping the same identification of variables as in $\psi$).
This decreases the number of connected components in $P$ by one, even in the case where $(k+1) \in C_1\cup C_2$.
Replacing $\psi$ in $\phi$ by $\psi'$ we obtain a new definition of $R$ with smaller sum of the number of the connected components of all the clauses and complete the induction step.

\case{2.2}{there exist $j\in \{1,2\}$ and $\bar{t}\in R$ witnessing items~\ref{item:building_bridges_1} and~\ref{item:building_bridges_2}}.
We fix any such $j\in \{1,2\}$ and $\bar{t}\in R$ for the rest of the proof. Let $\alpha$ be an arbitrary automorphism of $(\mathbb{Q};<)$ such that $\alpha \bar{t}\of{x_{i_j}}=0$.
Since $\phi$ is reduced, there is a tuple $\bar{s}\in R$  that does not satisfy the $i$-th disjunct in \eqref{eq:OrdHorn} for every $i\in [k+1]\setminus  C_{j}$.
Note that $\pp (\alpha \bar{t}, \bar{s}) \notin R$. 

\case{2.2.1}{$|C_1|=|C_{2}|=1$}. 
By our choice of $C_1$ and $C_2$, we have $|C|=1$ for every $C\in P$ and  $k+1\notin [k+1]\setminus C_1\cup C_2$.  
The relation $R'$ obtained from $R$ by conjunction with  equalities $ x_{i}=y_i$ 
and existentially quantifying all variables $x_{i},y_{i}$ for every $i\in [k+1]\setminus (C_{1}\cup C_2)$ is at most $4$-ary and since $\pp(\alpha \bar{t}, \bar{s}) \notin R$, $R'$ is not preserved by $\pp$ as well. This completes the induction step.

\case{2.2.2}{$|C_j|>1$}. Let $R'$ be the relation obtained by adding the constraint $(x_i=y_i)$ for some $i\in C_j\setminus \{k+1\}$ and then existentially quantifying over one of the variables $x_i,y_i$ that is not identified with $x_{k+1}$ or $y_{k+1}$; such a variable exists, otherwise $\psi$ is not reduced or holds always true. 
Then $R'$ is a relation of strictly smaller arity, which is pp-definable in $(\mathbb{Q};R)$. 
 Since $|C_j|>1$ and $\phi$ is reduced, we may assume that  $\bar{s}$ does not satisfy the $i$-th disjunct in \eqref{eq:OrdHorn} (it satisfies some other disjunct with an index in $C_j$).
 Recall that  $\pp(\alpha \bar{t}, \bar{s}) \notin R$,
 which entails that $R'$ is not preserved by $\pp$ as well. 
This completes the induction step.

\case{2.2.3}{$|C_j|=1$ and $|C_{3-j}|>1$}.  
Let $\ell\in C_{3-j}\setminus \{k+1\}$ be arbitrary. 
We now make another case distinction similar to the one for $\bar{t}$. This time it is  based on the (non)-existence of a tuple $\bar{t}'\in R$ witnessing items~\ref{item:building_bridges_1} and \ref{item:building_bridges_2} for the particular $j\in \{1,2\}$ from Case~2.2 and such that $\bar{t}'$  additionally satisfies $(x_{\ell}=y_{\ell})$.
We refer to this condition by $(\ast)$.

\case{2.2.3.1}{there exists $\bar{t}'\in R$ satisfying $(\ast)$}. 
Let $\alpha'$ be an arbitrary automorphism of $(\mathbb{Q};<)$ such that $\alpha' \bar{t}'\of{x_{i_j}}=0$.
Let $R'$ be a relation obtained by adding the constraint $(x_\ell=y_\ell)$ and existentially quantifying over one of the variables $x_\ell,y_\ell$ that is not identified with $x_{k+1}$ or $y_{k+1}$.
Note that $\pp(\alpha' \bar{t}', \bar{s}) \notin R$, so $R'$ is also not preserved by $\pp$.
Since $R'$ is a relation of strictly smaller arity and pp-definable in $(\mathbb{Q};R)$, 
this completes the induction step. 

\case{2.2.3.2}{no $\bar{t}'\in R$ satisfies $(\ast)$}.
Then $\phi$ entails 
\begin{align*}
 \Big(\bigvee\nolimits_{C\in P\setminus \{C_{3-j}\}}\NAE(x_{i},y_{i} \mid i \in C )\Big)    \vee    \big(x_{i_j}\geq \min(x_{i},y_{i} \mid i \in C_{3-j}) \big) \vee     (x_\ell\neq y_\ell)  
\end{align*}
if $k+1\in C_{3-j}$ and  
\begin{align*}
    & \Big(\bigvee\nolimits_{C\in P\setminus \{C_{3-j}\}}\NAE(x_{i},y_{i} \mid i \in C\setminus\{k+1\})\Big)  
   \vee    \big(x_{i_j}\geq \min(x_{i},y_{i} \mid i \in C_{3-j})\big) \\ 
 &  {}  \vee  (x_\ell\neq y_\ell)  \vee   (x_{k+1} \geq y_{k+1})  
\end{align*}
otherwise,   where we assume the same identification of variables as in $\psi$.
Let us denote the clause by $\psi'$.
 By Lemma~\ref{lemma:ElimMin}, we can simplify 
 $\psi'$ by replacing $\big(x_{i_j}\geq \min(x_{i},y_{i} \mid i \in C_{3-j})\big)$ with $(x_{i_j}\geq z)$ for some $z\in  \{x_{i},y_{i} \mid i \in C_{3-j}\}$.
We add  $\psi'$ to $\phi$ and reduce $\phi$; let $\psi''$ be the reduced version of $\psi'$. Note that the resulting formula still defines $R$. 
By the existence of $\bar{t}\in R$ and non-existence of $\bar{t}'$, $\psi''$ must contain the disjunct $(x_{\ell}\neq y_{\ell})$.
Since $\phi$ is reduced and $|C_{3-j}|>1$, there exists $\bar{r}\in R$ satisfying item~\ref{item:building_bridges_1} above Case~2.1 and $(x_{\ell}=y_{\ell})$. 
Hence, $\bar{r}$ satisfies $(x_{i_j} \geq z)$ and $\psi''$ must also contain the disjunct $(x_{i_j}\geq z)$.
Since $\psi''$ is reduced, there exists $\bar{s}' \in R$ such that it does not satisfy any disjunct of $\psi''$ except for $(x_{i_j}\geq z)$.

\case{2.2.3.2.1}{$i_j\neq k+1$}. Then the $i_j$-th disjunct in $\psi$
is of the form $(x_{i_j}\neq y_{i_j})$.
If the $i_j$-th disjunct in $\psi$ is also contained in $\psi''$, then we consider the relation $R'$ obtained from $R$ by adding the equality $(x_{i_j}=y_{i_j})$ and existentially quantifying the variable $y_{i_j}$.
Otherwise, there is no condition imposed on $y_{i_j}$ in $\psi''$ and we consider the relation $R'$ obtained from $R$ by only existentially quantifying the variable $y_{i_j}$. 
 Recall that $\alpha(\bar{t}\of{x_{i_j}})=0$ and observe that $\pp(\alpha \bar{t}, \bar{s}')$ does not satisfy $\psi''$ and hence is not in $R$. Hence, $R'$ is also not preserved by $\pp$ and of of strictly lower arity than $R$.
 This completes the induction step.  

\case{2.2.3.2.2}{$i_j= k+1$}.
 By Lemma~\ref{lemma:ElimMin} and since $\psi''$ is reduced, it cannot contain the disjuncts  $(x_{i_j}\geq z)$ and  $(x_{k+1}\geq y_{k+1})$ at the same time. Since $\psi'$ already contains the disjunct $(x_{i_j}\geq z)$, it cannot contain the disjunct $(x_{k+1}\geq y_{k+1})$
and therefore it does not impose any condition on the variable $y_{k+1}$.
Let $R'$ be the relation obtained from $R$ by existentially quantifying the variable $y_{k+1}$.
Then $R'$ is of strictly lower arity than $R$ and it is not preserved by $\pp$ since $\pp(\alpha \bar{t}, \bar{s}')$ does not satisfy $\psi''$ and hence is not in $R$. 
This completes the induction step.  
 \end{proof}

\subsection{Proof of Lemma~\ref{lem:notpp}\label{appendix:lemma:notpp}}

\NotPP*

Based on the catalogue of temporal relations in Table~\ref{table:catalogue}, we introduce the following notions. 
We say that that a ternary temporal relation $R$ is 
\begin{itemize}
\item  an \emph{M-relation}  if $\GMichal \subseteq R \subseteq \Michal$,
\item a  \emph{dual M-relation} if $\dualGMichal\subseteq R \subseteq \DMichal$,
\item a   \emph{strict M-relation} if $\GVSMichal \subseteq R \subseteq  \SMichal$,
\item a  \emph{dual strict M-relation} if $\dualGVSMichal \subseteq R \subseteq \DSMichal$,
\end{itemize}
We say that a quaternary temporal relation $R$ is 
\begin{itemize}
  \item a \emph{separated disjunction of disequalities} if $\LessSepGDis\subseteq R\subseteq \Dis$,
  \item a \emph{separated strict M-relation} if 
  \begin{itemize}
      \item either $\LessSepGSMichal\subseteq R$ or $\GreatSepGSMichal\subseteq R$, and
      \item $R\subseteq \SepSMichal$,
  \end{itemize}  
  \item a \emph{separated M-relation} if
  \begin{itemize}
      \item either $\LessSepGMichal\subseteq R$ or $\GreatSepGMichal\subseteq R$, and
      \item neither $R\subseteq \SepSMichal$ nor $R\subseteq \SepDima$, and
  \item    $R\subseteq \SepMichal$.
  \end{itemize}  
\end{itemize}  

 \begin{table}[t]   \caption{A catalogue of temporal relations for the proof of Lemma \ref{lem:notpp}.}
    \label{table:catalogue} \footnotesize
    \begin{tabular}{ll} \emph{Name} & \emph{Definition} \\ \toprule
 $\Dima(x_1, x_2, x_3) $   &  $ (x_1 \neq x_2 \vee x_2 = x_3)$ \\ \midrule
 $\SepDima(x_1, x_2, x_3,x_4)  $   &  $ (x_1 \neq x_2 \vee x_3 = x_4)$ \\ \midrule  
      $\Dis(x_1,x_2,x_3,x_4) $   & $  (x_1 \neq x_2 \vee x_3 \neq x_4)$ \\ \midrule  
       $\LessSepGDis(x_1, x_2, x_3, x_4) $   & $ (x_1 \neq x_2 \vee x_3 \neq x_4) \wedge (x_1 \leq x_2) \wedge (x_3 \leq x_4)  \wedge (\bigwedge\nolimits_{i,j \in \{ 1,2 \}}  x_i < x_{j+2})$ \\ \midrule  
        $\Michal(x_1, x_2, x_3)$   & $(x_1 \neq x_2 \vee x_2 \geq x_3)$ \\ \midrule  
   $\DMichal(x_1,x_2, x_3)$   & $(x_1 \neq x_2 \vee x_2 \leq x_3)$ \\ \midrule    
    $\SMichal(x_1,x_2, x_3) $   & $ (x_1 \neq x_2 \vee x_2 > x_3)$  \\ \midrule   
     $\DSMichal(x_1,x_2, x_3) $   & $(x_1 \neq x_2 \vee x_2 < x_3)$  \\ \midrule   
      $\GMichal(x_1, x_2, x_3) $   & $ (x_1 \neq x_2 \vee x_2 \geq x_3) \wedge (x_1 \geq x_2)$ \\ \midrule   
    $\dualGMichal(x_1, x_2, x_3) $   & $ (x_1 \neq x_2 \vee x_2 \leq x_3 ) \wedge (x_1 \leq x_2)$  \\ \midrule   
      $\GVSMichal(x_1, x_2, x_3)  $   & $  (x_1 \neq x_2 \vee x_2 > x_3) \wedge (x_1 \geq x_2) \wedge (\bigwedge_{i \in \{1,2\}}  x_i \neq x_3)$ \\ \midrule   
$\dualGVSMichal(x_1, x_2, x_3) $   & $ (x_1 \neq x_2 \vee x_2 <  x_3) \wedge (x_1 \leq x_2) \wedge (\bigwedge_{i \in \{1,2\}}  x_i \neq x_3)$  \\ \midrule      
      $\SepMichal(x_1,x_2,x_3,x_4) $   & $ (x_1\neq x_2\vee x_3\geq x_4)$  \\ \midrule    
      $\SepSMichal(x_1,x_2,x_3,x_4) $   & $ (x_1\neq x_2\vee x_3> x_4)$  \\ \midrule    
       $\LessSepGMichal(x_1, x_2, x_3, x_4) $   &  $(x_1 \neq x_2 \vee x_3 \geq x_4) \wedge (x_1 \geq x_2) \wedge  (\bigwedge\nolimits_{i,j \in \{ 1,2 \}}  x_i < x_{j+2} )$  \\ \midrule  
       $\GreatSepGMichal(x_1, x_2, x_3, x_4) $   &  $(x_1 \neq x_2 \vee x_3 \geq x_4) \wedge (x_1 \geq x_2) \wedge  (\bigwedge\nolimits_{i,j \in \{ 1,2 \}}  x_i >  x_{j+2} )$ \\ \midrule     
        $\LessSepGSMichal(x_1, x_2, x_3, x_4) $   & $ (x_1 \neq x_2 \vee x_3 > x_4) \wedge (x_1 \geq x_2)  \wedge (\bigwedge\nolimits_{i,j \in \{ 1,2 \}}  x_i < x_{j+2} )\wedge (x_3 \neq x_4) $ \\ \midrule  
      $\GreatSepGSMichal(x_1, x_2, x_3, x_4)   $   & $ (x_1 \neq x_2 \vee x_3 > x_4) \wedge (x_1 \geq x_2)   \wedge  (\bigwedge\nolimits_{i,j \in \{ 1,2 \}}  x_i >  x_{j+2} )\wedge (x_3 \neq x_4)$ \\ \bottomrule
    \end{tabular} 
\end{table}

\begin{proof}[Proof of Lemma~\ref{lem:notpp}] We first prove a short claim that providing us with tools for qpp-defining $\Zaneta$ using the relations introduced above.

 \begin{claim} \label{claim:short_tool}
The following statements hold:
\begin{enumerate}
\item \label{item1_trivial_observation} $\Michal$ pp-defines $\leq$ and qpp-defines $\neq$, $<$.
\item \label{item2_trivial_observation} A separated or dual M-relation together with $\neq$ pp-define a separated strict or dual strict M-relation, respectively. 
\item \label{item3_trivial_observation} A separated strict M-relation pp-defines a separated disjunction of disequalities.
\item \label{item4_trivial_observation} A dual strict M-relation together with $\Michal$ pp-define a separated disjunction of disequalities. 
\item \label{item5_trivial_observation} A separated disjunction of disequalities together with $\Michal$ qpp-define $\Zaneta$.
\end{enumerate}
\end{claim}

\begin{claimproof} 
Items~\ref{item1_trivial_observation} and~\ref{item2_trivial_observation} are straightforward to verify. 

For item~\ref{item3_trivial_observation}, let $R$ be a separated strict $M$-relation.
Then the pp-definition is given by
$
\exists a,b\ R(x_2,x_1,a,b)\wedge R(x_4,x_3,b,a).
$

For item~\ref{item4_trivial_observation}, let $R$ be a dual strict $M$-relation.
Then the pp-definition is given by 
$\exists h \ \Michal(x_1,y_1,h) \wedge R(x_2,y_2,h) \wedge  (x_1\leq x_2).$

For item~\ref{item5_trivial_observation}, observe first that using $<$, $\leq$, and a separated disjunction of disequalities, we can pp-define $\LessSepGDis$. We claim that, if $\Michal$ is available, then $(y_1 \neq x_1 \vee y_2 \neq x_2) \wedge (x_1 < x_2)$ is equivalent to 
$
\phi \coloneqq \exists v_1 \exists v_2~\LessSepGDis(x_1, v_1, x_2, v_2) \wedge \Michal(y_1, x_1, v_1) \wedge \Michal(y_2, x_2, v_2)
$.

It is easy to see that every tuple satisfying $(y_1 \neq x_1 \vee y_2 \neq x_2) \wedge (x_1 < x_2)$ also satisfies $\phi$. We show the converse direction. The claim is easy to see for each tuple satisfying $\phi\wedge (y_1 \neq x_1)\wedge (y_2 \neq x_2)$.
So consider instead a tuple satisfying  $\phi \wedge (y_1=x_1)$. Then, by the middle constraint, it satisfies $(x_1 \geq v_1)$.
By the first constraint we get $(x_1 = v_1)\wedge (x_2 < v_2)$, and then, by the third constraint, we get $(y_2 \neq x_2)$. The argument for $\phi \wedge (x_2 = y_2)$ is symmetric. 
\end{claimproof}

Next, we prove a long technical claim that essentially classifies OH relations that are not preserved by $\pp$ according to their shape. 
Note that the statement of Lemma~\ref{lem:notpp} follows directly from Claim~\ref{claim:long_boring_claim} combined with Claim~\ref{claim:short_tool}.

\begin{claim} \label{claim:long_boring_claim} If $\struct{B}$ is an OH structure that is not preserved by $\pp$, then it pp-defines one of the following:
    \begin{itemize}
    \item a dual M-relation, 
    \item a separated M-relation,
    \item a dual strict M-relation,  
    \item a separated strict M-relation, 
    \item a separated disjunction of disequalities,
    \item the relation $\Dima$.
    \end{itemize} 
\end{claim}
Let $\phi$ be an OH formula and let $\psi$ be an arbitrary clause in $\phi$. 
A \emph{subclause} of $\psi$ is a clause containing a subset of literals in $\psi$.
Typically, one views OH formulas over the signature $\{\neq,\geq\}$. Here, we
view them over the signature $\{\neq, >,=\}$, i.e., the disjunct $(x_{k+1}\geq y_{k+1})$ in \eqref{eq:CNF} stands for $(x_{k+1}>y_{k+1}) \vee (x_{k+1}=y_{k+1})$.
We do not need to view $(x_{i}\neq y_{i})$ as $(x_{i}<y_{i})\vee (x_{i}>y_{i})$, however, as this decomposition has no influence on the proof.
A subclause of a clause of the form \eqref{eq:CNF} may contain $(x_{k+1}>y_{k+1})$ or $(x_{k+1}=y_{k+1})$ and all other disjuncts.
Note that $(x_1\neq y_1 \vee \cdots \vee x_k\neq y_k \vee x_{k+1} > y_{k+1})$ is equivalent to
\begin{align*}
      (x_1\neq y_1 \vee \cdots \vee x_k\neq y_k \vee x_{k+1} \geq y_{k+1} )  \wedge  ( x_1\neq y_1 \vee \cdots \vee x_k\neq y_k \vee x_{k+1} \neq y_{k+1} ),
\end{align*} 
and $(x_1\neq y_1 \vee \cdots \vee x_k\neq y_k \vee x_{k+1} = y_{k+1})$ is equivalent to
\begin{align*}
     (x_1\neq y_1 \vee \cdots \vee x_k\neq y_k \vee x_{k+1} \geq y_{k+1} )  
 \wedge   ( x_1\neq y_1 \vee \cdots \vee x_k\neq y_k \vee y_{k+1} \geq x_{k+1} ).
\end{align*}  
Therefore, our convention about subclauses from above does not lead us out of the OH fragment.
We say that a clause is \emph{pretty} if it is preserved by $\pp$, and 
and  \emph{ugly} otherwise. 
Note that every clause containing at most two variables is pretty.
\begin{claimproof}  
By Lemma~\ref{lemma:arity_four}, we may assume that $\struct{B}$ contains only a single relation $R$ of arity $4$. 
In the following, let $\phi$ be the conjunction of all OH clauses entailed by $R(x_1, x_2, x_3, x_4)$;  clearly, there are only finitely many such clauses.
Note that, in contrast to the proof of Lemma~\ref{lemma:arity_four}, here $x_1$, $x_2$, $x_3$,  $x_4$
all represent different variables.
Without loss of generality, we may assume that $\phi$ contains no clauses of the form of the form $(x=y)$.
 Recall that we write that a tuple $\bar{t} \in R$   satisfies a formula $\psi(x_1,x_2,x_3,x_4)$ if $(\mathbb{Q};R)\models \psi(\bar{t}\of{1},\bar{t}\of{2}, \bar{t}\of{3}, \bar{t}\of{4})$.
In the present proof, we allow ourselves to also use instances of the constant $0$ in such formulas.
 
We close $\phi$ under the application of the five rules in Table~\ref{table:syntactic_pruning_rules},  
in the \emph{precise} order in which they appear.
Observe that the formula $\phi$ resulting from this procedure defines $R$.
 Also note that whenever $\phi$ contains a clause of the form $(x>y)$, it also contains a clause of the form $(x \geq y)$.
If $\phi$  contains no ugly clauses, then $R$ is preserved by $\pp$.
Otherwise, depending on the ugly clauses $\psi$ it contains we show that $R$ defines one of the relations specified in the formulation of the lemma. 
 We first consider clauses over three variables and then turn to clauses over four variables. Without loss of generality, we consider only reduced ugly clauses, as non-reduced clauses always entail their reduced version. In the case of ternary clauses we  assume without loss of generality that they are imposed on the first three variables. Moreover, observe that the permutation of variables does not matter.   

\begin{table}  
    \caption{Syntactic pruning rules for the proof of Lemma~\ref{lem:notpp}. All of the conditions on the right-hand side are to be understood as in a conjunction.}
    \label{table:syntactic_pruning_rules} 
    \scriptsize
     \thinmuskip=1mu
     \medmuskip=1mu
    \thickmuskip=1mu
    \begin{tabular}{rV{4cm}l}   
    & \emph{Action} & \emph{Condition(s)} \\
    \toprule
         I: & remove $(x_1 \neq x_2 \vee x_3 = x_4) $  & 
       \hspace{-0.5em}\begin{tabular}[t]{l} 1. $\phi$ contains a clause equivalent to $(x_1 \neq x_2 \vee x_2 = x_3) $    \\
           2. $\phi$ contains a clause equivalent to $(x_1 \neq x_2 \vee x_2 = x_4)$  \\ 
           \end{tabular}    \\ \midrule
II: & remove any ugly  $\psi$  &   $\phi$ entails a subclause of a clause equivalent to $\psi$   \\ \midrule
III: & replace $(x_1 \neq x_2 \vee x_3 > x_4)$ with $ (x_1 \neq x_2 \vee x_3 \neq x_4)$  & 
        $\phi$ contains the clause $(x_3 \geq x_4)$    \\ \midrule 
    IV: & remove $(x_1 \neq x_2 \vee x_2 = x_3)$ & 
       \hspace{-0.5em}\begin{tabular}[t]{l} 1. $\phi$ contains the clause $(x_1 \neq x_2 \vee x_2 \geq x_3)$   \\
           2.  $\phi$ contains the clause $(x_3 \geq x_1)$ or $(x_3 \geq x_2)$ \\ 
           \end{tabular}   \\ \midrule
V: &remove $(x_1 \neq x_2 \vee x_2 < x_3)$ & 
       \hspace{-0.5em}\begin{tabular}[t]{l} 1. $\phi$ contains the clause $(x_1 \neq x_2 \vee x_2 \neq x_3)$   \\
           2.  $\phi$ contains the clause $(x_3 \geq x_1)$ or $(x_3 \geq x_2)$ \\ 
           \end{tabular}   \\ \midrule
      VI: & remove any ugly  $\psi$ & \hspace{-0.5em}\begin{tabular}[t]{l} 1.  $\psi$ is a reduced clause with four distinct variables   \\
           2.  $\psi$ is entailed by a conjunction of pretty clauses in $\phi$ \end{tabular}  
         \\ \bottomrule      
    \end{tabular} 
\end{table}

\case{1}{$\psi$ is of the form $(x_1 \neq x_2 \vee x_2 = x_3)$}.   
Observe first that $\phi$ contains no subclauses of $\psi$, and hence $R$ contains a tuple satisfying $(x_1 = x_2 = x_3)$. Indeed, recall that $\phi$ was initially defined as the conjunction of all Horn clauses entailed by $R(x_1,x_2,x_3,x_4)$.  By Rule II in Table~\ref{table:syntactic_pruning_rules},  
$(x_1 \neq x_2)$ is not contained in $\phi$, as otherwise $\psi$ would have been removed from $\phi$.
As $(x_1 \neq x_2)$ is a pretty clause, it could not have been removed by any of the rules.
Therefore, $(x_1 \neq x_2)$ was not a clause in $\phi$ before the rules were applied and hence is not entailed by $\phi$.
Consequently, there is a tuple in $R$ satisfying $(x_1 = x_2)$.  Due to the presence of $\psi$, the very same tuple has to satisfy $(x_2 = x_3)$.   

Observe that $\phi$ entails $(x_1\neq x_2 \vee x_2\geq x_3)$. 
Suppose that $\phi$ entails $(x_1 \leq x_3)$.
Then both these clauses are present in $\phi$ before the rules in Table~\ref{table:syntactic_pruning_rules} are applied.
Since they are pretty clauses, they cannot be removed from $\phi$ by any of the rules in Table~\ref{table:syntactic_pruning_rules}.
But then $\psi$ is removed from $\phi$ using Rule IV in Table~\ref{table:syntactic_pruning_rules}, a contradiction.
Hence, $\phi$ does not entail $(x_1 \leq x_3)$.
Analogously, $\phi$ does not entail $(x_2 \leq x_3)$.

It follows that $R$ contains a tuple $\bar{t}_1$ satisfying $(x_1 > x_3)$, and a tuple $\bar{t}_2$ satisfying $(x_2 > x_3)$. 
Recall that all injective polymorphisms of $(\mathbb{Q}, \leq)$ preserve $R$ by Proposition~\ref{prop:ordhorn}. 
Note that $R$ contains a tuple $\bar{t}$ with positive entries satisfying $(x_1 \neq x_2)$ and $\elel(\bar{t},\bar{t}_2)$ satisfies $(x_1 \neq x_2)$ and $(x_2 > x_3)$. 
Therefore, we may  w.l.o.g. assume that $\bar{t}_2$ satisfies $(x_1 \neq x_2)$.
The application of $(x,y)\mapsto \lex(\max(x,y),\lex(x,y))$
to $\bar{t}_1, \bar{t}_2$ yields a tuple satisfying $(x_3 < x_1 < x_2)$ or $(x_3 < x_2 < x_1)$.
Without loss of generality assume that the first case holds.  

Let $R'$ be the projection of $R$ to the first three coordinates.
We consider two cases.
If $R'$ pp-defines $\leq$, then $\exists v~R'(x_1,x_2,v) \wedge (v \leq x_1)\wedge (x_1 \leq x_2) \wedge (v \leq x_3)$ specifies a relation pp-definable in $(\mathbb{Q};R)$. We show that this relation equals $\dualGMichal$, which is a dual M-relation.
Clearly, every tuple satisfying the formula lies in $\dualGMichal$. To see the reverse inclusion, consider an arbitrary tuple $\bar{r}\in \dualGMichal$. If $\bar{r}\of{1}=\bar{r}\of{2}$, then $v\coloneqq \bar{r}\of{1}$ certifies that $\bar{r}$ satisfies the formula. Otherwise $\bar{r}\of{1}<\bar{r}\of{2}$ and, using a suitable automorphism of $(\mathbb{Q}; <)$, we obtain a tuple $\bar{s}\in R'$ such that $\bar{s}\of{1}=\bar{r}\of{1}$, $\bar{s}\of{2}=\bar{r}\of{2}$, and $\bar{s}\of{3}<\min(\bar{r}\of{1},\bar{r}\of{3})$. 
Then $v\coloneqq \bar{s}\of{3}$ certifies that $\bar{r}$ satisfies the formula.

If $R'$ does not pp-define $\leq$, then it is easy to see that, for all distinct $i,j \in \{1,2,3\}$, the relation $R$ contains a tuple which satisfies $(x_i < x_j)$ and a tuple which satisfies $(x_j < x_i)$.
Since $R'$ is preserved by $\elel$ and $\dual\elel$,
it is straightforward to show that $R'$ contains all injective tuples.
Recall that a constant tuple is also in $R'$.
Hence, it follows that $R'(x_1, x_2, x_3) \wedge R'(x_1, x_3, x_2) \wedge R'(x_2, x_3, x_1)$ is equivalent to $(x_1 = x_2 = x_3) \vee (x_1 \neq x_2 \wedge x_2 \neq x_3 \wedge x_1 \neq x_3)$.
By Lemma~8.6 in~\cite{qecsps}, this formula specifies a relation that pp-defines $\Dima$.

\case{2}{$\psi$ is of the form $(x_1 \neq x_2 \vee x_2 < x_3)$}.  
Note that $\phi$ entails and hence contains the pretty clause $(x_1 \neq x_2 \vee x_1 \neq x_3)$.
Since $\psi$ is present in $\phi$, we may assume that $\phi$ entails neither $(x_1 \leq x_3)$ nor $(x_2 \leq x_3)$, otherwise $\psi$ would have been removed by Rule~V.
%
%
%
Hence neither $(x_1 \leq x_3)$ nor $(x_2 \leq x_3)$ is in $\phi$.
It follows that there is a tuple in $R$ satisfying $(x_3 < x_1)$, and perhaps a different one satisfying $(x_3 < x_2)$. 
As in Case~1, we obtain
a tuple which satisfies either $(x_3 < x_1 < x_2)$ or $(x_3 < x_2 < x_1)$. Without loss of generality we assume the former. By applying an automorphism of $(\mathbb{Q}; <)$ on this tuple, we obtain
a tuple $\bar{t}_1  \in R$   satisfying $(x_3 < x_1 <  0 < x_2)$. 
Since $x_1 \neq x_2$ is not entailed by $R$, there is also a tuple $\bar{t}_2 \in R$ satisfying $(x_1 = x_2 < 0 < x_3)$. The application $\elel(\bar{t}_2, \bar{t}_1)$ gives a tuple satisfying $(x_1 < x_2 < x_3)$, and $\dual\elel(\bar{t}_1,\bar{t}_2)$
a tuple satisfying $(x_1 < x_3 < x_2)$.
Observe now that the projection $R'$ of $R$ to the first three variables contains $\dualGVSMichal$, and hence is a dual strict M-relation.

\case{3}{$\psi$ is of the form $(x_1 \neq x_2 \vee x_2 \leq x_3)$}.  
 Note that $\psi$ is equivalent to $(x_1 \neq x_2 \vee x_1 \leq x_3)$.
By the construction of $\phi$, there is a tuple which satisfies $(x_1 = x_2 = x_3)$, a tuple which satisfies $(x_1 = x_2 < x_3)$, one which satisfies $(x_3 < x_1)$ and one which satisfies $(x_3 < x_2)$, otherwise $\psi$ would have been removed by Rule~2. 
As in Case~1, we may assume that the last tuple satisfies $(x_1 \neq x_2)$ and
obtain a tuple which satisfies $(x_3 < x_1 < x_2)$ or $(x_3 < x_2 < x_1)$.  Assume the former; the other case is symmetric.
Let $R'$ be the projection of $R$ to first three variables. Observe that $R'(x,x,y)$ defines $\leq$.
Now, as in the first case $\exists v~R'(x_1,x_2,v) \wedge (v \leq x_1 \leq x_2) \wedge (v \leq x_3)$ defines
$\dualGMichal(x_1, x_2, x_3)$, which is a dual M-relation. This completes the case.

\medskip From now on, we assume that ugly clauses of three variables do not occur in $\phi$, but it contains an ugly  clause on four variables. 

\case{4}{$\psi$ is of the form $(x_1 \neq x_2 \vee x_3 = x_4)$}.   
Recall that the five rules in Table~\ref{table:syntactic_pruning_rules} are applied exhaustively in the order in which they appear.
Suppose that $R$ contains no tuple satisfying $(x_1=x_2 \wedge x_2\neq x_3)$. Then, right before the application of the rules from Table~\ref{table:syntactic_pruning_rules}, $\phi$ would have contained the clause $(x_1\neq x_2 \vee x_2=x_3)$.
Since $\psi$ is a clause in $\phi$, every tuple satisfying $(x_1=x_2)$ also has to satisfy $(x_2=x_3=x_4)$.
But then $R$ contains no tuple satisfying $(x_1=x_2 \wedge x_2\neq x_4)$. Hence, right before the application of the rules from Table~\ref{table:syntactic_pruning_rules}, $\phi$ would have contained the clause $(x_1\neq x_2 \vee x_2=x_4)$.
By Rule I in Table~\ref{table:syntactic_pruning_rules}, $\psi$ would have been removed from $\phi$, a contradiction.
Thus, it must be the case that $R$ contains a tuple $\bar{s}_1$ satisfying $(x_1=x_2 \wedge x_2\neq x_3)$.
Since $\psi$ is a clause in $\phi$, it also satisfies $(x_3=x_4)$.

Now we distinguish the two cases where $\bar{s}_1$ satisfies either $(x_1=x_2<x_3=x_4)$ or $(x_1=x_2>x_3=x_4)$.
Since $\phi$  
contains no clauses of the form $(x=y)$, in both cases, the relation $R$ has a tuple satisfying $(x_1 \neq x_2)$ and a tuple satisfying $(x_3 \neq x_4)$. The application of $\lex$ to both tuples yields a tuple $\bar{s}_2  \in R$ that satisfies both $(x_1 \neq x_2)$ and $(x_3 \neq x_4)$. 
Now, up to transposing the entries $x_1,x_2$ and/or $x_3,x_4$ in $R$, the tuple $\bar{s} = \lex(\bar{s}_1, \bar{s}_2)$ satisfies $(x_1 < x_2 < x_3 < x_4)$ if $\bar{s}_1$ satisfies $(x_1=x_2<x_3=x_4)$ or $(x_1 > x_2 > x_3 > x_4)$ if $\bar{s}_1$ satisfies $(x_1=x_2>x_3=x_4)$.

Assume first that $R$ pp-defines $\leq$. If $\bar{s}$ satisfies $(x_1 < x_2 < x_3 < x_4)$,
then the formula
\begin{align*}
    \exists v \ R(x_1, x_2, x_3, v) \wedge \left(\bigwedge\nolimits_{i,j \in \{ 1,2 \}} x_i \leq x_{j+2} \right)    \wedge (x_1 \leq v)\wedge (x_2 \leq v) \wedge (x_4 \leq v)
\end{align*} 
pp-defines a separated M-relation. 
Indeed, it is straightforward to check that it entails $\SepMichal(x_1,x_2, x_3, x_4)$, but neither $\SepSMichal(x_1, x_2, x_3, x_4)$ nor $\SepDima(x_1, x_2, x_3, x_4)$.  
Moreover, it is entailed by $\LessSepGMichal(x_1, x_2, x_3, x_4)$.
 If $\bar{s}$ satisfies $(x_1 > x_2 > x_3 > x_4)$, then the formula
\begin{align*}
   \exists v \ R(x_1, x_2,v, x_4) \wedge \left( \bigwedge\nolimits_{i,j \in \{ 1,2 \}} x_i \geq x_{j+2} \right)  
   \wedge  (x_1 \geq v)\wedge (x_2 \geq  v)  \wedge (x_3 \geq v) 
\end{align*}  
 entails $\SepMichal(x_1, x_2, x_3, x_4)$. On the other hand, the formula entails neither $\SepSMichal(x_1, x_2, x_3, x_4)$ nor
 $\SepDima(x_1, x_2, x_3, x_4)$. Moreover, it is entailed by $\GreatSepGMichal(x_1, x_2, x_3, x_4)$. Hence, the formula defines a separated $M$-relation.
 
If $<$ is pp-definable, then, in a similar way, one shows that 
\begin{align*}
    \exists v \ R(x_1, x_2, x_3, v) \wedge \left(\bigwedge\nolimits_{i,j \in \{ 1,2 \}} x_i < x_{j+2} \right)    \wedge (x_1 < v)\wedge (x_2 < v) \wedge (x_4 < v)
\end{align*}  
pp-defines a separated strict M-relation in the case where  $\bar{s}$ satisfies $(x_1 < x_2 < x_3 < x_4)$ and that
 \begin{align*}
   \exists v \ R(x_1, x_2,v, x_4) \wedge \left( \bigwedge\nolimits_{i,j \in \{ 1,2 \}} x_i > x_{j+2} \right)  
   \wedge  (x_1 > v)\wedge (x_2 >  v)  \wedge (x_3 > v) 
\end{align*}   
  pp-defines a separated strict M-relation in the case where $\bar{s}$ satisfies $(x_1 > x_2 > x_3 > x_4)$.

The remaining subcase is when neither $\leq$ nor $<$ are pp-definable. 
Let $\psi'$ be an arbitrary clause in $\phi$ containing a disjunct $(x_i \lhd x_j)$ with $\lhd\in \{\leq,<\}$.  
If $R$ does not contain any tuple that does not satisfy any disjunct in $\psi'$ besides $(x_i \lhd x_j)$, then we can remove $(x_i \lhd x_j)$ from $\psi'$.
So suppose that $R$ contains such a tuple. Since $<$ is not pp-definable, it follows that $\lhd$ equals $\leq$, otherwise we could pp-define $<$ just by adding equality conjuncts and projecting onto the entries $\{i,j\}$.
Since $\leq$ is not pp-definable, by the same reasoning as above, it cannot be the case that there are two tuples $\bar{t},\bar{t}'\in R$ which do not satisfy any disjunct in $\psi'$ besides $(x_i \leq x_j)$, and additionally $\bar{t}$ satisfies $(x_i=x_j)$ and $\bar{t}'$ satisfies $(x_i<x_j)$.
Hence, one of the tuples is missing in $R$.
Since $<$ is not pp-definable, it cannot be the case that $\bar{t}$ is missing, so it must be $\bar{t}'$.
But note that then we can replace $\psi'$ by the subclause obtained containing $(x_{i}=x_{j})$ instead of $(x_{i}\leq x_{j})$.
%
%
Applying this observation finitely many times yields an OH definition of $R$ which does not use any $\{<,\leq\}$-atoms, i.e., $R$ is an equality relation.

Recall that $R$ contains a tuple satisfying $(x_1<x_2<x_3<x_4)$ or $(x_1>x_2>x_3>x_4)$.
Since $R$ is an equality relation, it is preserved by all automorphisms of $(\mathbb{Q};=)$, and hence contains all injective tuples.
Suppose that $R$ contains a constant tuple.
Let $R'$ be the $4$-ary relation defined by $$ \exists u,v \left(\bigwedge\nolimits_{i,j\in [4], i\neq j} R(x_i,x_j,u,v) \wedge R(u,v,x_i,x_j)\right).
$$  
Since every tuple in $R$ satisfies $\psi$, $R'$ consists of all injective tuples and all constant tuples.
But then  the formula $\exists x_4 \ R'(x_1,x_2,x_3,x_4)$ defines the relation $S$ from Lemma~8.6 in \cite{qecsps}, and hence $R$ pp-defines $\Dima$ by the same lemma.

Finally, suppose that $R$ does not contain any constant tuple.
Then, by Lemma~4 in \cite{bodirsky2008complexity}, $R$ pp-defines $\neq$.
Recall that $R$ contains a tuple satisfying $(x_1=x_2\neq x_3=x_4)$.
Let $R''$ be the $4$-ary relation defined by $R(x_1,x_2,x_3,x_4)\wedge R(x_3,x_4,x_1,x_2)$.
Since $R$ is preserved by all automorphisms of $(\mathbb{Q};=)$, $R''$ contains all tuples satisfying $(x_1=x_2\neq x_3=x_4)$.
Since $R''$ also contains all injective tuples and every tuple in $R''$ satisfies $(x_1=x_2 \Leftrightarrow x_3=x_4)$, the formula $$R''(x_1,x_2,x_3,x_4) \wedge \left(\bigwedge\nolimits_{i,j\in \{1,2\}} x_i\neq x_{j+2}\right)$$ defines the relation $T$ from Lemma~8.7 in \cite{qecsps}.
By the same lemma, $R$ pp-defines $\Dima$. 

\case{5}{$\psi$ is of the form $(x_1 \neq x_2 \vee x_3 \neq x_4)$}.    
Observe that $\psi$ entails $\theta \coloneqq (x_1 \neq x_2 \vee x_1 \neq x_3 \vee x_1 \neq x_4)$, which is a pretty clause. By the construction of $\phi$, the formula $\theta$ is contained in $\phi$.

 We now argue that $R$ contains
\begin{itemize}
\item both a tuple $\bar{s}_1$ satisfying
$(x_1 = x_2 < x_3) \wedge (x_1 = x_2 <  x_4)$  
and a tuple $\bar{s}_2$ satisfying  $(x_1< x_3 = x_4) \wedge (x_2 < x_3 = x_4)$,
\item  or both a tuple $\bar{s}'_1$ satisfying $(x_3 = x_4 < x_1) \wedge (x_3 = x_4 < x_2)$
and a tuple $\bar{s}'_2$ satisfying $(x_3 < x_1 = x_2)\wedge (x_4 < x_1 = x_2)$.
\end{itemize}

Suppose it is not the case. There are altogether  four cases to consider.
Before we step into the details, we observe the following.
  
\begin{enumerate}
    \item \label{claim:sepneqs1} If $R$ does not contain $\bar{s}_1$, then $\phi$ contains $\psi_1 \coloneqq (x_1 \neq x_2 \vee x_1 \geq x_i)$ for some $i \in \{ 3, 4 \}$,
    \item \label{claim:sepneqs2} If $R$ does not contain $\bar{s}_2$, then $\phi$ contains $\psi_2 \coloneqq (x_k \lhd x_l)$ for some $k \in \{ 3,4\} , l \in \{ 1,2 \}, \lhd \in \{ \leq, < \}$.
    \item \label{claim:sepneqs3} If $R$ does not contain $\bar{s}'_1$, then $\phi$ contains $\psi'_1 \coloneqq (x_3 \neq x_4 \vee x_3 \geq x_i)$ for some $i \in \{ 1, 2 \}$.
    \item \label{claim:sepneqs4} If $R$ does not contain $\bar{s}'_2$, then $\phi$ contains $\psi'_2 \coloneqq (x_l \lhd x_k)$ for some $k \in \{ 3,4\} , l \in \{ 1,2 \}, \lhd \in \{ \leq, < \}$.
\end{enumerate}
 Note that the operation $f(x,y)\coloneqq \lex(\max(x,y), \lex(x,y))$ is binary, injective, and preserves $\leq$. Hence, by Proposition~\ref{prop:ordhorn}, it preserves $R$.
 
    If $R$ does not contain $\bar{s}_1$, then it either does not contain a tuple satisfying 
    $(x_1 = x_2 < x_3)$ or a tuple  satisfying $(x_1 = x_2 <  x_4)$, otherwise an application of $f$
    to these tuples would yield a tuple $\bar{s}_1$. 
    Hence  $\phi$ entails the pretty clause $(x_1 \neq x_2 \vee x_2 \geq x_i)$ for some $i \in \{3,4\}$. By the construction of $\phi$ it follows that the clause is in the formula. 

If $R$ does not contain $\bar{s}_2$, then, similarly, by the invariance under $f$, 
either a tuple satisfying $(x_1<x_3=x_4)$ or a tuple satisfying $(x_2<x_3=x_4)$ is not in $R$. Hence
$\phi$ entails  $(x_3 \neq x_4 \vee x_3 \leq x_i)$ for some $i \in \{1,2\}$. 
By Case~3, we may assume that such a clause is not in $\phi$, otherwise we are done. 
Therefore, this clause must have been removed from $\phi$ by Rule II from Table \ref{table:syntactic_pruning_rules}, because $\phi$ contained a subclause $\psi'$ of a clause equivalent to $\psi$. Either $\psi'$ is an ugly clause on three variables, which had to be removed due to Rule II, IV or V, or $\psi'$ is a clause on two variables. In each of the cases we eventually obtain that $\phi$ contains a subclause on 2 variables of a clause equivalent to $\psi$: either because it is a subclause that caused removal of a clause on three variables by Rule II, or because it was required for application of Rule IV or V. Since $\psi$ was not removed and $\phi$ contains no conjuncts of the form $(x=y)$, we have that $\phi$ contains a clause $(x_k \lhd x_l)$ from Item~\ref{claim:sepneqs2}.

The reasoning for the two remaining items is analogous.  

Now if $R$ contains neither $\bar{s}_1$ nor $\bar{s}'_1$, we have both $\psi_1$ and $\psi'_1$ in $\phi$. It yields the contradiction with Rule VI in the construction of $\phi$ since
$\psi_1 \wedge \psi'_1 \wedge \theta$ entails $\psi$ and all the three clauses 
$\theta, \psi_1, \psi'_1$ are pretty. 

If $R$ contains neither $\bar{s}_1$ nor $\bar{s}'_2$, then $\psi_1$ and $\psi_2'$ are in $\phi$. We need to consider two cases, namely, whether $i = k$ or not.
 If $i=k$, then $\phi$ entails   $(x_1 \neq x_2 \vee x_1 = x_i)$ for some $i \in \{ 3, 4 \}$. 
Recall that we assume that no ugly clause with three variables is present in $\phi$.  
Therefore, $(x_1 \neq x_2 \vee x_1 = x_i)$ must have been removed by Rule II or IV in the table in Table~\ref{table:syntactic_pruning_rules}.
If it was removed by Rule II, then it contradicts the assumption that $\phi$ contains no equality conjuncts or the presence of $\psi$ in $\phi$ (using Rule II again).
Suppose that it was removed by Rule IV.
Then $\phi$ contains the pretty clauses $(x_1\neq x_2 \vee x_{1}\geq x_i)$, $(x_i\geq x_{1})$ or $(x_i \geq x_2)$, and $\theta$.
In any case, the conjunction of these clauses entails $\psi$. Therefore, $\psi$ would have been removed by Rule VI, a contradiction.
On the other hand, if 
$i \neq k$, then
$\psi_1 \wedge \psi'_2 \wedge \theta$ entails $\psi$. Since all
$\psi_1, \psi'_2$ and $\theta$ are pretty, $\psi$ cannot be in $\phi$ because of Rule VI. 

The case where $R$ contains neither $\bar{s}'_1$ nor $\bar{s}_2$ is symmetric to the previous one. 

 The remaining case is where $R$ contains neither $\bar{s}_2$ nor $\bar{s}'_2$. 
Then $\phi$ contains $\psi_2 \coloneqq (x_{k_1} \lhd_1 x_{l_1})$ and $\psi'_2 \coloneqq (x_{l_2} \lhd_2 x_{k_2})$ for some $k_1, k_2 \in \{3,4 \}$, $l_1, l_2 \in \{ 1,2 \}$ and $\lhd_1, \lhd_2 \in \{ \leq, < \}$.
Observe that if $k_1 = k_2$ and $l_1 = l_2$, then $\phi$ must contain $(x_{k_1} = x_{l_1})$, which contradicts the fact that $\phi$ contains no equalities.  On the other hand, if both $k_1, k_2$ and $l_1, l_2$ are pairwise different, we have that  $\psi_2 \wedge \psi'_2 \wedge \theta$ entails $\psi$, so $\psi$ would have been removed by Rule VI.
Therefore, consider the case where $l_1 \neq l_2$ and $k_1 = k_2$; the case where $l_1 = l_2$ and $k_1 \neq k_2$ is similar. In this case $\psi_2 \wedge \psi'_2$ entails 
$(x_{l_1} \neq x_{l_2} \vee x_{l_1} = x_{k_1})$. Hence, we have that 
$\psi_2 \wedge \psi'_2 \wedge \theta$ entails $\psi$, a contradiction to the construction of $\phi$.

From now on we may assume that $R$ either contains both $\bar{s}_1, \bar{s}_2$ or both
$\bar{s}'_1, \bar{s}'_2$. Without loss of generality assume that $R$ contains 
$\bar{s}_1$ which satisfies $(x_1 = x_2 < x_3 < x_4)$ and $\bar{s}_2$ which satisfies $(x_1 < x_2 < x_3 = x_4)$;  the remaining cases may be achieved by transposing $x_1$ with $x_2$, $x_3$ with $x_4$, or the pairs $(x_1, x_2)$ with $(x_3,x_4)$. Moreover $\lex(\bar{s}_1,\bar{s}_2) \in R$ satisfies $(x_1<x_2<x_3<x_4)$.
Now it is easy to verify that $\phi(x_1,x_2,x_3,x_4)$ defines a separated disjunction of disequalities.
%


\case{6}{$\psi$ is of the form $(x_1 \neq x_2 \vee x_3 > x_4)$}.     
We first observe that $\phi$ does not entail the pretty clause $(x_3 \geq x_4)$.
Indeed, if it were the case, then $(x_3 \geq x_4)$ would be a clause of $\phi$ and by Rule III, $\psi$ would have been replaced by $(x_1 \neq x_2 \vee x_3 \neq x_4)$. 
Since $(x_3 \geq x_4)$ is pretty and was therefore not removed from $\phi$ by any of the rules, $R$ must contain a tuple satisfying $(x_3<x_4)$.
Since $\phi$ contains no subclauses of $\psi$, the relation $R$ contains a tuple satisfying $(x_1=x_2 \wedge x_4<x_3)$.
Thus, we may assume that the projection of $R$ to the entries $\{x_3,x_4\}$ contains all injective pairs.

Now either $R$ contains a tuple $\bar{s}_1$ satisfying $(x_1 = x_2 < x_4 < x_3)$, or $\phi$ entails and by the construction also contains the pretty clause $\psi' \coloneqq (x_1 \neq x_2 \vee x_1 \geq x_4)$. 
Furthermore, if $\psi'$ is present, then either the relation $R$ contains  a tuple $\bar{s}_2$ satisfying $(x_1 = x_2 > x_3 > x_4)$, or $\phi$ entails $(x_1 \neq x_2 \vee x_1 \leq x_3)$.
We show that $R$ contains $\bar{s}_1$ or $\bar{s}_2$. Suppose that this is not the case.
Since we assume no ugly clauses on three variables are present in $\phi$, the entailed clause $(x_1 \neq x_2 \vee x_1 \leq x_3)$ must have been removed by Rule II.
Therefore, there is a subclause of a clause equivalent to $(x_1 \neq x_2 \vee x_1 \leq x_3)$ entailed by $\phi$. This subclause is either pretty or removed by Rule II, IV or V, which yields in all cases that $\phi$ contains a pretty clause $\psi''$ that entails $(x_1 \neq x_2 \vee x_1 \leq x_3)$. 
Observe now that $\psi$ entails $\theta \coloneqq (x_1 \neq x_2 \vee x_1 \neq x_3 \vee x_1 \neq x_4)$ and that $\theta \wedge \psi' \wedge \psi''$ entails $\psi$. Since  all $\theta, \psi'$ and $\psi''$ are pretty, 
this contradicts the presence of $\psi$ in $\phi$.
It follows that $R$ contains either $\bar{s}_1$ or $\bar{s}_2$.

Since $R$ entails no equality conjuncts, it contains a tuple $\bar{t}_1$ satisfying $(x_1 \neq x_2)$. 
Recall that $R$ also contains a tuple $\bar{t}_2$ satisfying $(x_3 < x_4)$.
We may assume that $\bar{t}_1$ satisfies $(x_1 > x_2)$; the other case follows by transposing $x_1$ and $x_2$.
If $R$ contains $\bar{s}_1$, then $\bar{t}\coloneqq \lex(\bar{s}_1,\bar{t}_1)\in R$ satisfies $(x_3>x_4>x_1>x_2)$.
W.l.o.g., $\bar{t}$ satisfies $(x_4>0>x_1)$, otherwise replace $\bar{t}$ with a tuple obtained by an application of a suitable automorphism of $(\mathbb{Q};<)$ to $\bar{t}$. 
Consequently, the tuple $\elel(\bar{t},\bar{t}_2)\in R$ satisfies $(x_1>x_2)$, $(x_3<x_4)$, and $(x_1<x_3)$.
But this means that $R$ contains $\LessSepGSMichal$, and thus $R$ is a separated strict M-relation.
If $R$ contains $\bar{s}_2$, then $\bar{t}\coloneqq \lex(\bar{s}_2,\bar{t}_1)\in R$ satisfies $(x_1>x_2>x_3>x_4)$.
W.l.o.g., $\bar{t}$ satisfies $(x_2>0>x_3)$.
Consequently, the tuple $\dual\elel(\bar{t},\bar{t}_2)\in R$ satisfies $(x_1>x_2)$, $(x_3<x_4)$, and $(x_2>x_4)$.
But this means that $R$ contains $\GreatSepGSMichal$, and thus $R$ is a separated strict M-relation.

\case{7}{$\psi$ is of the form $(x_1 \neq x_2 \vee x_3 \geq x_4)$}.   
Note that $\phi$ does not entail $(x_1 \neq x_2 \vee x_3 = x_4)$;  otherwise $\psi$ would have been removed by Rule II.

Analogously, $\phi$ does not entail $(x_1 \neq x_2 \vee x_3 > x_4)$.
Hence, $R$ contains a tuple satisfying $(x_1=x_2)\wedge (x_3>x_4)$, as well as a tuple satisfying $(x_1=x_2)\wedge (x_3=x_4)$.
It follows that $\exists h \ R(h,h,x_1,x_2)$ defines $\geq$.
Since $R$ is OH and pp-defines $\geq$, it follows from Theorem~20 in \cite{bodirsky2010complexity} that either 
$R$ pp-defines $<$, or $R$ has a constant polymorphism.
Indeed, the fact that OH structures cannot pp-define any of the three relations Sep, Cycl, Betw from Theorem~20 in \cite{bodirsky2010complexity} follows directly from Proposition~\ref{prop:ordhorn} because none of the relations is preserved by $\elel$. 

Clearly, $\phi$ does not contain the clause $(x_3\geq x_4)$,  otherwise $\psi$ would not be present in $\phi$. Hence, $R$ contains a tuple $\bar{t}$ satisfying $(x_1 \neq x_2)$ and $(x_3 < x_4)$.
Suppose that $\bar{t}$ satisfies  $(x_1 > x_2)$;  the case $(x_1 < x_2)$ follows by transposing the entries $x_1$ and $x_2$.
 
We claim that $R$ contains a tuple $\bar{s}_1$ satisfying $(x_1=x_2<x_4)$ or a tuple $\bar{s}_2$ satisfying $(x_1=x_2>x_3)$.
Suppose, on the contrary, that $R$ contains neither $\bar{s}_1$ nor $\bar{s}_2$.
Then $R(x_1,x_2,x_3,x_4)$ entails $(x_1\neq x_2\vee x_1 \geq x_4)$ and $(x_1\neq x_2\vee x_1 \leq x_3)$.
Since $(x_1\neq x_2\vee x_1 \geq x_4)$ is pretty, it is contained in $\phi$.
Since $(x_1\neq x_2\vee x_1 \leq x_3)$ is ugly and we assume that $\phi$ does not contain any ternary ugly clauses, it must have been removed using Rule II.
As in Case~6, $\phi$ must contain a pretty clause $\psi'$ which entails $(x_1\neq x_2\vee x_1 \leq x_3)$.
%
But note that $(x_1\neq x_2\vee x_1 \geq x_4)$ together with
$\psi'$
entails $\psi$.
Therefore, $\psi$ would have been removed using Rule VI, a contradiction.
Hence $R$ contains $\bar{s}_1$ or $\bar{s}_2$ .

Suppose that $R$ contains $\bar{s}_1$.
W.l.o.g.,  $\bar{s}_1$ satisfies $(x_4>0>x_1)$, otherwise replace $\bar{s}_1$ with a tuple obtained by an application of a suitable automorphism of $(\mathbb{Q};<)$ to $\bar{s}_1$. 
Then $\elel(\bar{s}_1,\bar{t})\in R$ satisfies $(x_1>x_2)$, $(x_3<x_4)$, and $(x_1<x_3)$.
Recall that $R$ contains a tuple $\bar{r}$ satisfying $(x_1=x_2)$ and $(x_3>x_4)$.
Then $\elel(\elel(\bar{s}_1,\bar{t}),\bar{r})$ satisfies $(x_1>x_2)$, $(x_3>x_4)$, and $(x_1<x_4)$.

If $R$ has no constant polymorphism, then, as argued above, $R$ pp-defines $<$.
In this case, the formula $\exists h \ R(x_1,x_2,x_3,h)\wedge (h>x_4)$ defines a separated strict M-relation.  
Otherwise $R$ has a constant polymorphism.
If $R$ contains a tuple $\bar{s}$ satisfying $(x_1=x_2<x_3=x_4)$, then $R$ pp-defines a separated M-relation.
Indeed, note that $\lex(\bar{s},\bar{r})\in R$ satisfies $(x_1=x_2<x_4<x_3)$.
%
%
It is easy to verify that the relation $R'$ defined by
$
\exists h  \ R(x_1,x_2,x_3,h)\wedge  R(x_1,x_2,h,x_4)
$
contains all tuples from $R$, and additionally a tuple satisfying $(x_2<x_1<x_3=x_4)$ because $R$ contains both a tuple satisfying $(x_2<x_1<x_3<x_4)$ and a tuple satisfying $(x_2<x_1<x_4<x_3)$. Therefore, $R'$ contains $\LessSepGMichal$ and is a separated M-relation.

Otherwise,  recall that $R$ pp-defines $\leq$.
Consider the ternary relation $R'$ defined by $ \exists h \ R(h,h,x_2,x_1) \wedge (x_3\geq h) \wedge (h\leq x_1)$.
Note that, if the first two entries in a tuple from $R'$ are equal, then the third entry must be greater than or equal, because $R$ does not contain any tuple satisfying $(x_1=x_2<x_3=x_4)$.
Therefore, $R'\subseteq \DMichal$.
We claim that also $\dualGMichal\subseteq R'$.
Indeed, note that the tuple $\bar{s}_1\in R$ satisfies $(x_1=x_2<x_4<x_3)$, witnessing all tuples in $\dualGMichal$ where the first two entries are not equal.
Moreover, the constant tuple in $R$ witnesses all tuples in $\dualGMichal$ where the first two entries are equal.
Hence, $R'$ is a dual $M$-relation. 

Suppose that $R$ contains $\bar{s}_2$, without loss of generality $\bar{s}_2$ satisfies $x_1 > 0 > x_3$.
Then the argumentation is almost entirely symmetrical to the case where $R$ contains $\bar{s}_1$, except that we use $\dual\elel$ instead of $\elel$ and $\GreatSepGMichal$ instead of $\LessSepGMichal$.
The only case that needs to be handled separately is when $R$ has a constant polymorphism and does not contain any tuple satisfying $(x_1=x_2>x_3=x_4)$.
In this case, the tuple $\dual\elel(\bar{s}_2,\bar{t})$ satisfies $(x_1>x_2>x_4>x_3)$ and hence it easy to see that the formula
$ \exists a,b \ R(x_2,x_1,a,b) \wedge (x_2\geq x_1) \wedge (x_1\geq b) \wedge (b\geq a) \wedge    (x_3\geq a)$
pp-defines a dual $M$-relation.

\medskip  
In the final part of the proof, assume that $\phi$ does not contain any of the clauses in Cases~1--7. What remains are the clauses $\psi$ of the form 
$(x_4 \neq x_1 \vee x_1 \neq x_2 \vee x_2 \lhd x_3)$ or the form $(x_4 \neq x_1 \vee x_1 \rhd x_2 \vee x_2 \neq x_3)$ where $\lhd\in \{ =, <, \leq,  \}$ and $\rhd\in \{ =, >, \geq  \}$.
Indeed, whenever two disjuncts in an OH clause share the same variables, we can either remove one of them, or we can remove the whole clause because it is trivially true.
Note that, any reduced OH clause $\psi$ with free variables $\{x_1,x_2,x_3,x_4\}$ contains at most three disequality disjuncts.
If it is exactly three, then the disequality part is already equivalent to $\NAE(x_1,x_2,x_3,x_4)$.
If the clause contained another disjunct $x \rhd y$, then either it would not be reduced or it would be trivially true. Otherwise, $\psi$ is equivalent to $\NAE(x_1,x_2,x_3,x_4)$ and hence a pretty clause.
Therefore, any reduced OH clause with at most four free variables contains at most three disjuncts.

In all these cases, we will consider the ternary relation $R'$ defined by $R(x_1, x_2, x_3, x_1)$ together with the OH definition $\phi'$ of $R'$ obtained as follows.
We start with $\phi'$ defined as a conjunction of all OH clauses that are entailed by $R'$.
Next, we close $\phi'$ under the application of the five rules in Table \ref{table:syntactic_pruning_rules} to $\phi'$ in the order in which they appear. We use $R'$ with its OH definition $\phi'$ to pp-define one of the relations from the statement of the lemma.

Note that $\phi'$ may in general contain clauses of the form $(x=y)$.
We argue that, if $\psi$ is of the form $(x_4 \neq x_1 \vee x_1 \neq x_2 \vee x_2 \lhd x_3)$, then $\phi'$ does not contain the clause $(x_1=x_2)$.
Suppose for contradiction that $\phi'$ contains $(x_1=x_2)$. Then $\phi$ entails $(x_4 \neq x_1 \vee x_1 = x_2)$, which, together with $\psi$, entails $(x_4 \neq x_1 \vee x_2 \lhd x_3)$, which is a subclause of $\psi$.
Since $\psi$ was not removed by Rule II, this clause must have been removed from $\phi$ by Rule I. In particular, $\lhd$ is $=$.
The two ugly clauses with two disjuncts that triggered application of Rule I cannot be present in $\phi$ either and thus are entailed by pretty clauses (Rule~II or Rule~IV). 
Therefore, $\psi$ would have been removed by Rule VI, a contradiction.  
Analogous, even simpler, argument shows that if $\psi$ is of the form $(x_4 \neq x_1 \vee x_1 \rhd x_2 \vee x_2 \neq x_3 )$, then $\phi'$ does not contain a clause $(x_2 = x_3)$. This suffices for applying arguments from Cases~1--3, which are the only cases from above that we refer to in Cases~8--13.

\case{8}{$\psi$ is of the form $(x_4 \neq x_1 \vee x_1 \neq x_2 \vee x_2 = x_3)$}.   
We aim to show that $\phi'$ contains $\psi' = (x_1 \neq x_2 \vee x_2=x_3)$, then by the same reasoning as in Case~1 we obtain a pp-definition of one of the relations in the statement of the lemma. Suppose that $\psi'$ is not present in $\phi'$. Since $\psi'$ is entailed by $R'$, it must have been removed by Rule II or IV. Since $\psi$ is present in $\phi$ and $\psi'$ is a subclause of $\psi$, it must have been Rule IV. Hence, $\phi'$ contains $(x_3 \geq x_i)$ for some $i \in \{1,2\}$ and therefore $\phi$ entails $(x_4 \neq x_1 \vee x_3 \geq x_i)$. By our assumption that $\phi$ does not contain this ugly clause or its ugly subclauses, this clause must be entailed by a conjunction $\theta$ of pretty clauses in $\phi$. Moreover, $\phi$ entails the pretty clause $(x_4 \neq x_1 \vee x_1 \neq x_2 \vee x_1 \geq x_3)$.
Now it is easy to see that $\theta \wedge (x_4 \neq x_1 \vee x_1 \neq x_2 \vee x_1 \geq x_3)$
entails $\psi$, and hence $\psi$ should not be in $\phi$, a contradiction. Therefore $\psi'$ is present in $\phi'$ as we wanted to prove.

\case{9}{$\psi$ is of the form $(x_4 \neq x_1 \vee x_1 \neq x_2 \vee x_2 < x_3)$}. Then $\psi' = (x_1 \neq x_2 \vee x_2 < x_3)$ is entailed by $\phi'$ and it could not have been removed by Rule II, because then $\psi$ would have been removed from $\phi$ as well. Suppose that it was removed by Rule V. 
As in the previous case, there must be a conjunction $\theta$ of pretty clauses in $\phi$ entailing $(x_4 \neq x_1 \vee x_3 \geq x_i)$ for some $i \in \{1,2\}$. Moreover, $\phi$ entails the pretty clause $(x_1 \neq x_4 \vee x_1 \neq x_2 \vee x_1 \neq x_3)$, which in conjunction with $\theta$ entails $\psi$. This contradicts the presence of $\psi$ in $\theta$, hence $\psi'$ is present in $\phi'$. By the same reasoning as in Case~2 we pp-define from $R'$ one of the relations from the statement of the lemma.

\case{10}{$\psi$ is of the form $(x_4 \neq x_1 \vee x_1 \neq x_2 \vee x_2 \leq x_3)$}. Then $\psi' = (x_1 \neq x_2 \vee x_2 \leq x_3)$ is entailed by $\phi'$ and it could not have been removed by Rule II, because then $\psi$ would have been removed from $\phi$ as well. Therefore $\psi'$ is present in $\phi'$ and by the same reasoning as in Case~3 we pp-define one of the relations from the statement of the lemma.

\case{11}{$\psi$ is of the form $(x_4 \neq x_1 \vee x_1 > x_2 \vee x_2 \neq x_3)$}. 
Then $\psi' = (x_1 > x_2 \vee x_2 \neq x_3)$ is entailed by $\phi'$ and it could not have been removed by Rule II, because then $\psi$ would have been removed from $\phi$ as well. Suppose that it was removed by Rule V. 
Then $\phi'$ contains $(x_1 \geq x_i)$ for some $i \in \{2,3\}$ and hence $\phi$ entails the pretty clause $(x_4 \neq x_1 \vee x_1 \geq x_i)$ for some $i \in \{2,3\}$. Moreover, $\phi$ entails the pretty clause $(x_1 \neq x_4 \vee x_1 \neq x_2 \vee x_1 \neq x_3)$, which in conjunction with $(x_4 \neq x_1 \vee x_1 \geq x_i)$  entails $\psi$. This contradicts the presence of $\psi$ in $\theta$. Therefore, $\psi'$ is present in $\phi'$ and we may apply the reasoning from Case~2 to pp-define from $R'$ one of the relations from the statement of the lemma.

\case{12}{$\psi$ is of the form $(x_4 \neq x_1 \vee x_1 \geq x_2 \vee x_2 \neq x_3)$}.
Then $\psi'=(x_1 \geq x_2 \vee x_2 \neq x_3)$  is entailed by $\phi'$ and it could not have been removed by Rule II, because then $\psi$ would have been removed from $\phi$ as well. Therefore $\psi'$ is present in $\phi'$ and by the same reasoning as in Case~3 we pp-define one of the relations from the statement of the lemma.

\medskip
The last remaining case is when $\psi$ is of the form $(x_4 \neq x_1 \vee x_1=x_2 \vee x_2 \neq x_3)$. In this case, we also use the ternary relation $R''$ defined by $R(x_1, x_2, x_2, x_4)$ and its OH definition $\phi''$ that arises analogously to the definition $\phi'$ of $R'$. Recall that $\phi'$ does not contain a clause $(x_2=x_3)$. By symmetry, $\phi''$ does not contain a clause $(x_4=x_1)$.

\case{13}{$\psi$ is of the form $(x_4 \neq x_1 \vee x_1 = x_2 \vee x_2 \neq x_3)$}. It is enough to show that $\psi' = (x_1 = x_2 \vee x_2 \neq x_3)$ is present in $\phi'$ or that $\psi'' = (x_4 \neq x_1 \vee x_1 = x_2)$ is present in $\phi''$, then the reasoning from Case~1 can be applied to $R'$ or $R''$, respectively, to pp-define one of the relations from the statement of the lemma. Suppose for contradiction that neither $\psi'$ is present in $\phi'$ nor $\psi''$ is present in $\phi''$. Since $\psi'$ is entailed by $R'$, it was removed by Rule II or IV. Clearly, it could not have been Rule II, otherwise $\psi$ would not be present in $\phi$ as well. Therefore it must have been Rule IV and $\phi'$ contains $(x_1 \geq x_i)$ for some $i \in \{2,3\}$. Hence, $\phi$ entails and therefore contains the pretty clause $\theta' = (x_4 \neq x_1 \vee x_1 \geq x_i)$. 
Analogously, $\psi''$ was removed from $\phi''$ by Rule IV and $\phi''$ contains $(x_2 \geq x_j)$ for some $j \in \{1,4\}$, which in turn implies that $\phi$ contains the pretty clause $\theta'' = (x_2 \geq x_j \vee x_2 \neq x_3)$. However $\theta' \wedge \theta''$ entails $\psi$, contradicting the presence of $\psi$ in $\phi$. Therefore, $\psi'$ is present in $\phi'$ or $\psi''$ is present in $\phi''$ as we wanted to prove.
\end{claimproof}
This concludes the proof of the lemma.
\end{proof}

\end{document}